\documentclass[a4paper,onecolumn,11pt,accepted=2024-05-30]{quantumarticle}
\pdfoutput=1
\usepackage{lipsum}

\usepackage[utf8]{inputenc}  
\usepackage[T1]{fontenc}     
\usepackage[english]{babel}  
\usepackage{hyperref}  
\usepackage{graphicx} 
\usepackage[babel]{microtype}  
\usepackage{amsmath,amssymb,amsthm,bm,amsfonts,mathrsfs,bbm} 
\usepackage{setspace}
\usepackage{braket}
\usepackage[bottom]{footmisc}
\usepackage{blochsphere}
\usepackage{pgfplots}
\pgfplotsset{compat=1.18}
\usetikzlibrary{calc,3d,shapes, pgfplots.external, intersections}
\usepackage{algorithm}
\usepackage[noend]{algorithmic}
\usepackage{eqparbox}

\usepackage{subfig}
\usepackage{xfrac}
\usepackage{tikz}
\usepackage{mathtools}
\usepackage{mathabx}
\usetikzlibrary{arrows.meta}
\usepackage{csquotes}

\usepackage{xspace}  
\usepackage{pgfplots}
\usepackage{verbatim}
\usepackage{amsmath}
\usepackage[shortlabels]{enumitem}
\DeclareMathOperator*{\argmax}{arg\,max}
\DeclareMathOperator*{\argmin}{arg\,min}

\newtheorem{theorem}{Theorem}

\newcommand{\Lagr}{\mathcal{L}}

\newcommand\id{\leavevmode\hbox{\small1\kern-3.3pt\normalsize1}}

\newcommand{\bs}[1]{\boldsymbol{#1}}
\newcommand{\bst}{\bs{\theta}}
\newcommand{\bsy}{\bs{y}}
\newcommand{\bsx}{\bs{x}}
\newcommand{\bsz}{\bs{0}}
\newcommand{\bse}{\bs{e}}
\newcommand{\lV}{\left\Vert}
\newcommand{\rV}{\right\Vert}

\newcommand{\mbbR}{\mathbb{R}}

\newtheorem{definition}[theorem]{Definition}

\newtheorem{lemma}[theorem]{Lemma}

\newtheorem{assumption}{Assumption}

\allowdisplaybreaks

\begin{document}

\title{Variational Quantum Algorithms for Semidefinite Programming}

\author{Dhrumil Patel}
\affiliation{Department of Computer Science, Cornell University, Ithaca, New York, 14850, USA}
\email{djp265@cornell.edu}
\orcid{0000-0003-2433-0314}

\author{Patrick J. Coles}
\affiliation{Theoretical Division, Los Alamos National Laboratory, Los Alamos, New Mexico 87545, USA}
\email{patrick@normalcomputing.ai}
\homepage{http://patcoles.com/}
\orcid{0000-0001-9879-8425}

\author{Mark M. Wilde}
\affiliation{Hearne Institute for Theoretical Physics, Department of Physics and Astronomy,
and Center for Computation and Technology, Louisiana State University, Baton Rouge, Louisiana 70803, USA}
\affiliation{School of Electrical and Computer Engineering,
Cornell University,
Ithaca, New York 14850, USA}
\email{wilde@cornell.edu}
\homepage{http://markwilde.com/}
\orcid{0000-0002-3916-4462}

\maketitle

\begin{abstract}
A semidefinite program (SDP) is a particular kind of convex optimization problem with applications in operations research, combinatorial optimization, quantum information science, and beyond. In this work, we propose variational quantum algorithms for approximately solving SDPs. 
For one class of SDPs, we provide a rigorous analysis of their convergence to approximate locally optimal solutions, under the assumption that they are weakly constrained (i.e., $N\gg M$, where $N$ is the dimension of the input matrices and  $M$ is the number of constraints). 
We also provide algorithms for a more general class of SDPs that requires fewer assumptions.
Finally, we numerically simulate our quantum algorithms for applications such as MaxCut, and the results of these simulations provide evidence that convergence still occurs in noisy settings.

\end{abstract}

\tableofcontents

\section{Introduction}

\label{sec:introduction}

Semidefinite programming (SDP) is  one of the most important tools in optimization developed over
the past few decades. One can see SDPs as a natural extension of the better known linear programs (LPs), in which the vector inequalities of LPs are replaced by matrix inequalities.
One of the reasons underlying the
importance of SDPs is their applicability to a broad range of problems, including approximation algorithms for combinatorial optimization~\cite{LS91}, control theory~\cite{majumdar2020recent}, and sum-of-squares~\cite{Parrilo2003}. Additionally, a variety of quantum information problems can be formulated as SDPs, including state discrimination~\cite{YKL75,Eld03}, upper bounds on quantum channel capacity~\cite{WXD18,W18thesis}, and self-testing~\cite{SB20}.

The power of SDPs lies with the fact that they can be solved efficiently in polynomial time using classical algorithms such as the
celebrated interior-point method~\cite{PW2000}. Although SDPs can be solved efficiently using classical techniques, as the size of the input matrices increases, many first-order and second-order algorithms incur significant computational overhead due to the expensive gradient computation at each iteration. For this reason, it is imperative to design more efficient algorithms for solving SDPs.

Given the speed-ups of quantum algorithms over classical algorithms for a variety of problems~\cite{Sho94, HHL09, Gro97, DH96, HV03}, it is natural to ask if there exists a quantum algorithm that can solve SDPs efficiently. This question was positively answered  in~\cite{BS16}, wherein  a quantum algorithm was proposed and proven to have a quadratic speedup over the classical Arora--Kale algorithm~\cite{AK16}. Following this initial result, more efficient quantum algorithms for solving SDPs were later developed~\cite{QASDP19,AGGW17}.

Although quantum algorithms have been theoretically proven to outperform known classical algorithms for many applications, a fault-tolerant quantum computer is required to reap their benefits. Currently, fault-tolerant quantum computers are not available, and we are instead in the Noisy Intermediate-Scale Quantum (NISQ) era~\cite{Pre18}. Google constructed a noisy quantum computer with just 54 qubits~\cite{Goo19}, but it remains an open challenge to design a fault-tolerant quantum computer that requires millions of qubits for successful operation. Some of the quantum information science research community is focused on determining the power of noisy quantum computers and is designing quantum algorithms that acknowledge the limitations of such devices \cite{VQA21,bharti2021noisy,Chen2023,angrisani:tel-04511706}, such as a finite number of gates and qubits, noisy gate execution, and rapid decoherence of qubits.

Variational quantum algorithms (VQAs) constitute an important class of NISQ-friendly algorithms~\cite{VQA21,bharti2021noisy}. VQAs can be seen as hybrid quantum-classical algorithms that have a classical computer available for optimization, only calling a quantum subroutine for tasks that are not efficiently solvable by it. Hitherto, VQAs have been proposed for numerous computational tasks, for which there are also known quantum and classical algorithms~\cite{VQA21}. Some well studied VQAs include the Variational Quantum Eigensolver (VQE)~\cite{VQE14} and the Quantum Approximate Optimization Algorithm (QAOA)~\cite{FGG14}.

This brings us to the main motivation of our work, where we reformulate three different kinds of SDPs and develop variational quantum algorithms for solving them. In particular, the contributions of our paper are as follows:
\begin{enumerate}
    \item We present unconstrained reformulations of the constrained general and standard forms of an SDP, by employing a series of reductions (Sections~\ref{sec:reform-gen-SDP}, \ref{sec:ecsf}, and \ref{sec:ICSF-reform}). For the standard form, we consider the case in which the SDP is weakly constrained (i.e., when $N \gg M$, where $N$ is the dimension of the input matrices and  $M$ is the number of constraints). On the contrary, we do not require such assumptions when considering the general form, making it applicable to a wide range of problems.
    
    \item We propose variational quantum algorithms for obtaining approximate stationary points of such unconstrained problems, which in turn are approximate stationary points of their respective constrained problems due to the equivalence between them (Sections~\ref{sec:vqagf}, \ref{sec:vqa-ecsf}, and \ref{sec:vqa-icsf}). Note that stationary points include globally optimal points.
    
    \item We analyze the convergence rate of an algorithm corresponding to the equality constrained standard formulation of an SDP (Section~\ref{sec:conanaleq}).
    
    \item We provide numerical evidence that showcases how our algorithms work in practice for applications such as MaxCut (Section~\ref{sec:simulations}). Specifically, we analyze the convergence of the proposed algorithms by assessing how close the final value is to the actual optimal value evaluated using exact classical solvers. We perform such experiments on a noisy quantum simulator from Qiskit~\cite{Qiskit} (i.e., QASM simulator) and then compare these results with those on a noiseless simulator. Note that the QASM simulator mimics IBM's quantum computer, which is actually noisy in nature due to gate errors and decoherence.
\end{enumerate}

Additionally, in Section~\ref{sec:pre}, we introduce some notations and definitions, the basics of SDPs, and a brief overview of VQAs, as well as the associated concept of computing partial derivatives on a quantum circuit, i.e., the parameter-shift rule.

\subsection{Main Idea and Setup}

In this paper, we consider the following three different kinds of SDPs, and we present variational quantum algorithms for all three of them:
\begin{itemize}
    \item \textit{General Form (GF)}: The general form of an SDP can be concisely written as follows:
    \begin{equation}\label{eq:genprimalcon}
p^* \coloneqq \sup_{X\succcurlyeq0}\left \{  \operatorname{Tr}[CX]:\Phi(X)\preccurlyeq B \right\},
\end{equation}
where $C \in \mathcal{S}^{N}$, $B \in \mathcal{S}^{M}$, and the map $\Phi$ is  Hermiticity-preserving (Definition~\ref{def:hermilmap}). Here, the notation $\mathcal{S}^{N}$ denotes the set of $N \times N$ Hermitian operators. For this case, we do not assume that the SDPs are weakly constrained; i.e., we do not assume that the dimension $M$ of the constraint variable $B$ is much smaller than the dimension $N$ of the objective variable~$C$. Additionally, to generalize it further, we consider an inequality-constrained problem. For solving these general SDPs, we propose a variational quantum algorithm (Algorithm~\ref{algo:iVQAGF}).

\item \textit{Standard Form (SF)}: Here, we consider the Hermiticity-preserving map $\Phi$ and the Hermitian operator $B$ to have a diagonal form, as given in~\eqref{eq:sdpgeneralform}. Specifically, in this case, we set
\begin{equation}
    \begin{aligned}
        \Phi(X) & = \text{diag}\left( \operatorname{Tr}[A_{1}X],\ldots, \operatorname{Tr}[A_{M}X] \right),\\
        B &= \text{diag}\left (b_{1},\ldots,b_{M} \right),
    \end{aligned}
\end{equation}
    where $A_{1}, \ldots, A_{M} \in \mathcal{S}^{N}$ and $b_{1}, \ldots, b_{M}\in \mathbb{R}$. This form of an SDP is well known in the convex optimization literature~\cite{BV04}, and most exact or approximation algorithms designed for combinatorial optimization problems are based on this form. 
    For this case, we assume that the SDPs are weakly constrained, i.e., $N \gg M$. We further categorize them based on the nature of the constraints as follows:
    \begin{itemize}
        
        \item \textit{Equality Constrained Standard Form (ECSF)}: Here, we consider equality constraints, i.e., $\Phi(X) = B$. For solving such SDPs, we propose a variational quantum algorithm (Algorithm~\ref{algo:iVQAEC}) and establish its convergence rate and total iteration complexity.
        
        \item \textit{Inequality Constrained Standard Form (ICSF)}: Here, we consider inequality constraints, i.e., $\Phi(X) \preccurlyeq B$, and we propose a variational quantum algorithm for this case (Algorithm~\ref{algo:iVQAIC}).
        
    \end{itemize}
\end{itemize}

In our methods, we first reduce these constrained optimization problems to their unconstrained forms by employing a series of identities. Second, we express the final unconstrained form as a function of expectation values of the input Hermitian operators, i.e., $C, A_{1},\ldots, A_{M} \in \mathcal{S}^{N}$.

For solving these final unconstrained formulations using a gradient-based method, we need access to the full gradient of the objective function at each iteration. However, when the dimension $N$ of the input Hermitian operators is large, the evaluation of the gradient of the objective function using a classical computer becomes computationally expensive. Therefore, we delegate this gradient computation to parameterized quantum circuits, and we use a technique called the parameter-shift rule~\cite{Li2017,Mitarai2018,SBG19} to evaluate the partial derivatives of the objective function with respect to circuit parameters. Please refer to Section~\ref{sec:vqa} for more details on the parameter-shift rule.

  We design variational quantum algorithms for solving these unconstrained optimization problems. Our methods provide bounds on the optimal values, due to the reduction in the search space, as well as the non-convex nature of the objective function landscape in terms of quantum circuit parameters. We do not assume that the final objective function is convex with respect to these parameters, as generally it is non-convex~\cite{HD21}. 
 In general, finding a globally optimal point for a non-convex function is known to be NP-hard~\cite[Section 2.1]{Danilova2022},  and so an important question regards the type of solutions that we can guarantee in such a scenario. In the classical optimization literature, the notion of approximate stationary points (Definition~\ref{def:estationary}) is considered when the objective function is non-convex. Therefore, in this paper, we focus on proving the convergence of our algorithms to approximate stationary points, as proving the same for a global optimal point is quite difficult. We would also like to emphasize that VQAs that have been proposed prior to our paper have not considered this notion, which is quite natural when considering non-convex objective functions.

Intuitively, stationary points are those for which the gradient of the function under consideration is equal to zero. Therefore, a stationary point can be a local maximum (including the global maximum), a local minimum (including the global minimum), or a saddle point. Consequently, when using a first-order solver such as gradient descent, one may get stuck at one of these points. However, it is often desirable to escape unwanted stationary points and move to a more favorable one, depending on the type of optimization problem (maximization or minimization). This can be achieved either through the use of higher-order solvers or by employing a noisy first-order solver~\cite{Jin0NKJ17}. We focus on the latter approach in this paper, as computing higher-order derivatives can be computationally expensive compared to calculating the first-order gradient. Moreover, due to the inherent stochastic nature of variational quantum algorithms (VQAs), a first-order solver can inherently be noisy. This noise can aid in escaping unwanted stationary points and converging to a better stationary point.

 The primary reason behind proposing different variational quantum algorithms for three different kinds of SDPs is the nature of the final unconstrained forms of these SDPs. In the general form of an SDP, as we do not consider any assumptions on the input matrices, the final unconstrained optimization problem turns out to be a non-convex--non-concave optimization problem (see~\eqref{eq:genfinalform}). Proving the convergence of the proposed algorithm to an approximate stationary point of this problem requires sophisticated analysis, which we leave for future work. In contrast, although the standard form of SDPs is a special case of the general form, we consider cases where these SDPs are weakly constrained. Studying the standard form of SDPs with this assumption is important because many SDPs of interest are actually large and weakly constrained. Due to this assumption, we observe that the final unconstrained forms turn out to be non-convex--concave optimization problems (see~\eqref{eq:auglageqpara} and~\eqref{eq:form2paramsoft}). Non-convex--concave optimization~\cite{daskalakis2018training, heusel2018gans, mertikopoulos2018optimistic, mazumdar2019finding, rafique2021weaklyconvex} has a rich literature when compared to that of non-convex--non-concave optimization,  as the latter is a harder problem than the former. Therefore, the design and convergence analysis of more sophisticated variational quantum algorithms for solving non-convex--concave optimization problems is possible.

 \textbf{Reformulation of General Form (GF) of SDPs}\textemdash For the SDPs written as~\eqref{eq:genprimalcon}, we reduce this form to the following final unconstrained form expressed in terms of quantum circuit parameters $\bst_{1}\in [0, 2\pi]^{r_{1}}$ and $\bst_{2}\in [0, 2\pi]^{r_{2}}$:
\begin{equation}\label{eq:genfinalform}
    p^* \coloneqq \sup_{\substack{\bst_{1}  \in [0, 2\pi]^{r_{1}},\\\lambda\geq0}}\inf_{\substack{\bst_{2}  \in [0, 2\pi]^{r_{2}},\\\mu\geq
0}} \left\{\lambda \langle  I \otimes C^{\top} \rangle_{\bst_{1}}+\mu\langle I \otimes B \rangle_{\bst_{2}}\\ -\lambda\mu\langle  I \otimes I \otimes \Gamma^{\Phi} \rangle_{\bst_{1},\bst_{2}}\right \},
\end{equation}
where $\Gamma^{\Phi}$ is the Choi operator (Definition~\ref{def:choiop}) of the linear map $\Phi$. See Section~\ref{sec:reform-gen-SDP} for more details. For solving the above optimization problem, we propose a VQA called \textit{inexact Variational Quantum Algorithm for General Form} (iVQAGF), in which we have two parameterized quantum circuits competing against each other to maximize/minimize the objective function, and then there is a classical optimizer that updates the parameters of these quantum circuits.

\textbf{Reformulation of Equality Constrained Standard Form (ECSF) of SDPs}\textemdash As stated before, for this type of SDP, we make an assumption on it being weakly-constrained, i.e., $N \gg M$. By exploiting this assumption, we design more sophisticated variational quantum algorithms, in which we need just one parameterized quantum circuit. For such SDPs, we arrive at the following final unconstrained form expressed in terms of quantum circuit parameters~$\bst\in [0, 2\pi]^{r}$:
\begin{equation}
  p^*  \coloneqq \sup_{\bs{\theta} \in [0, 2\pi]^{r}}\inf_{\bs{y}\in \mbbR^{M}}\left \{ \lambda \left \langle I \otimes C \right \rangle_{\bs{\theta}} + \bs{y}^\top \left(\bs{b} - \lambda \bs{\Phi}(\bs{\theta})\right) + \frac{c}{2} \left \Vert \bs{b} - \lambda \bs{\Phi}(\bs{\theta})\right \Vert^{2} \right\}. \label{eq:auglageqpara}
\end{equation}
See Section~\ref{sec:ecsf} for more details. Throughout this paper, we use the notation
\begin{equation}
\langle H\rangle_{\bst} \equiv \langle \phi(\bst)| H |\phi(\bst)\rangle
\end{equation}
to represent the expectation value of a Hermitian operator $H$ with respect to $|\phi(\bst)\rangle$. 
For solving the above optimization problem, we propose a VQA called \textit{inexact Variational Quantum Algorithm for Equality Constrained standard form} (iVQAEC). We run iVQAEC  on a classical computer, and at any step of the algorithm, the expectation value of a Hermitian operator is evaluated using a parameterized quantum circuit. One of our main results establishes the convergence rate and total iteration complexity of iVQAEC under the assumption of the SDP being weakly-constrained.

\textbf{Reformulation of Inequality Constrained Standard Form (ICSF) of SDPs}\textemdash For this case also we make the weakly-constrained assumption on SDPs. Here, we solve the dual problem instead of the primal problem. For this type of SDP, we arrive at the following final unconstrained form expressed in terms of quantum circuit parameters $\bst\in [0, 2\pi]^{r}$:
\begin{align}\label{eq:form2paramsoft}
d' \coloneqq  \sup_{\bs{\theta} \in [0, 2\pi]^{r}} \inf_{\bar{\bs{y}} \geq 0}\ \left \{ \sum_{i=1}^{M-1}b_{i}y_{i}+
\frac{b_{M}}{\gamma}\ln \!\left ( e^{\gamma \langle I \otimes H(\bar{\bs{y}})\rangle_{\bst}} + 1\right )\right\}.
\end{align}
See Section~\ref{sec:ICSF-reform} for more details. For solving this problem, we propose a VQA called \textit{inexact Variational Quantum Algorithm for Inequality Constrained standard form} (iVQAIC). For this algorithm, we do not establish its convergence rate, which we leave for future work; instead, we prove a property of the objective function that is necessary for providing such a convergence analysis, i.e., smoothness of the objective function for a fixed $\bar{\bs{y}}$.

\subsection{Related Work}

Recently, an approach to semidefinite programming on NISQ devices was proposed in~\cite{BHVK21}, which is non-variational  and called NISQ SDP solver (NSS) therein. We should note that this approach does not provide an efficient solution for a general SDP problem. Like our approach, the NSS approach also optimizes over a subset of the positive-semidefinite operator space. The NSS approach assumes the ability to prepare pure states in a set $\left\{|\psi_{i}\rangle\right\}_i$. Using these states, the NSS approach then constructs a hybrid density matrix ansatz:  $X_{\beta} = \sum_{i, j} \beta_{ij} |\psi_{i}\rangle \langle \psi_{j}|$, where the entries $\{\beta_{ij}\}_{ij}$ are stored on a classical computer. The NSS approach then transforms the original SDP into a low-dimensional SDP, which is solved by optimizing over these classical entries. Note that the NSS approach does not change the set of pure states for each iteration. As one can see, the NSS approach does not encompass the entire space of positive semidefinite operators; therefore, it is heuristic in nature, similar to ours. This means that there is no guarantee of convergence to the global optimal point.

Additionally, another technique was proposed in~\cite{MA21} where the authors ``quantized'' the classical randomized cutting plane method for solving semidefinite programs. They used a quantum eigensolver subroutine in order to speedup the classical method. Their results indicate that the robustness of their approach against noise may be useful in implementing their method on NISQ devices.

Our approach here is complementary to both of these approaches because our algorithm is a variational quantum algorithm. Another important point to note here is that, unlike prior works that focus solely on solving SDPs in the standard form, we propose algorithms for solving SDPs in both the standard form and a form considered in~\cite{Wat18,KW20}. The latter form is prevalent in quantum information theory, as it is used to compute various relevant quantities like fidelity and trace distance. While one can convert this form to the standard form, working in the original form is more convenient, avoiding unnecessary conversion overhead. Comparing the aforementioned approaches to ours is an interesting direction for future work.

\section{Preliminaries}

\label{sec:pre}

In this section, we introduce some notations and definitions.

\textbf{Notations:} We denote the set of real and complex numbers by $\mathbb{R}$ and $\mathbb{C}$, respectively. For a positive integer $m$, the notation $[m]$ denotes the set $\{1,\ldots,m\}$. We use upper-case letters to denote matrices and bold lower-case letters to denote vectors (e.g., $A$ is a matrix, and $\boldsymbol{x}$ is a vector). Let $\mathcal{H}$ denote a finite-dimensional Hilbert space of dimension $N$. This Hilbert space represents a system of $n$ qubits, where $n = \lceil \log_{2} N \rceil$. The notation $L(\mathcal{H})$ denotes the set of linear operators acting on $\mathcal{H}$. The set of $N\times N$ Hermitian or self-adjoint matrices is denoted by $\mathcal{S}^{N} \subset L(\mathcal{H})$. The notation $\mathcal{S}^{N}_{+} \subset \mathcal{S}^{N}$ denotes the set of $N\times N$ positive semidefinite (PSD) matrices, and $\mathcal{D}^{N} \subset \mathcal{S}^{N}_{+}$ denotes the set of $N\times N$ density matrices. 
Additionally, we use $\left \Vert \cdot \right \Vert$ to represent the $\ell_{2}$ norm of a vector, as well as the spectral or operator norm of a matrix (its largest singular value), and it should be clear from the context which is being used. 
Let $\operatorname{Tr}[X]$ denote the trace of a matrix $X$, i.e., the sum of its diagonal terms. Let $X^\top$ and $X^{\dagger}$ denote the transpose and Hermitian conjugate (or adjoint) of the matrix $X$, respectively. The notation $A \succcurlyeq B$ or $A-B \succcurlyeq 0$ indicates that $A-B \in \mathcal{S}^{N}_{+}$. For a multivariate function $f : \mathbb{R}^{n} \rightarrow \mathbb{R}$, we use $\nabla f$ and $\nabla^{2} f$ to denote its gradient and Hessian, respectively. Let $\frac{\partial f (\cdot)}{\partial x_{i}}$ denote the partial derivative of $f$ with respect to $i^{\text{th}}$ component of the vector $\boldsymbol{x}$. For a multivariate vector-valued function $f$ : $\mathbb{R}^{n} \rightarrow \mathbb{R}^{m}$, we denote its Jacobian by $J_{f}(\cdot)$. 

Let $\mathcal{H}$ and $\mathcal{H}'$ be Hilbert spaces of dimensions $N$ and $N'$, respectively.

\begin{definition}[Linear map]\label{def:lmap}
A map $\Phi : L(\mathcal{H})\rightarrow  L(\mathcal{H}')$ is a linear map if the following holds $\forall X$, $Y \in L(\mathcal{H})$ and $\forall \alpha, \beta \in \mathbb{C}$,
\begin{equation}\label{eq:defmap1}
    \Phi(\alpha X + \beta Y) = \alpha \Phi(X) + \beta \Phi(Y).
\end{equation}
\end{definition}

\begin{definition}[Adjoint of a linear map]\label{def:adlmap}
The adjoint $\Phi^{\dagger} : L(\mathcal{H}') \rightarrow L(\mathcal{H})$ of a linear map $\Phi : L(\mathcal{H})\rightarrow  L(\mathcal{H}')$ is the unique linear map such that
\begin{equation}\label{eq:defmap2}
    \langle Y, \Phi(X)\rangle = \langle \Phi^\dagger(Y), X\rangle \quad \forall X \in L(\mathcal{H}),\ \forall Y \in L(\mathcal{H}'),
\end{equation}
where $\langle C,D\rangle \coloneqq \operatorname{Tr}[C^\dag D]$ is the Hilbert--Schmidt inner product.
\end{definition}

\begin{definition}[Hermiticity-preserving linear map]\label{def:hermilmap}
A linear map $\Phi : L(\mathcal{H}) \rightarrow L(\mathcal{H}')$ is a Hermiticity-preserving map if $\Phi(X)\in \mathcal{S}^{N'}$ for all $X\in \mathcal{S}^{N}$. Equivalently, $\Phi$ is Hermiticity preserving if and only if $\Phi(X^{\dagger}) = \Phi(X)^{\dagger}$ for all $ X \in L(\mathcal{H})$.
\end{definition}

\begin{definition}[Choi representation of a linear map] \label{def:choiop}
For every linear map $\Phi : L(\mathcal{H})\rightarrow  L(\mathcal{H}')$, its Choi representation $\Gamma^{\Phi}$ is defined as
\begin{equation}\label{eq:choiop}
    \Gamma^{\Phi}\coloneqq \sum_{i,j=0}^{N-1}|i\rangle \langle j|\otimes\Phi(|i\rangle \langle
j|).
\end{equation}
The operator $\Gamma^{\Phi}$ is also known as the Choi operator.
Additionally, for every linear map $\Phi$, the following holds $\forall X \in L(\mathcal{H}), Y \in L(\mathcal{H}')$:
\begin{equation}
    \operatorname{Tr}[Y\Phi(X)] = \operatorname{Tr}[(X^\top\otimes Y)\Gamma^{\Phi}],
\end{equation}
which is a direct consequence of the well known fact that
\begin{equation}
\Phi(X) = \operatorname{Tr}_{1}[{\Gamma^{\Phi}(X^\top\otimes I)}],
\label{eq:choi-implement-map}
\end{equation}
 with the partial trace over the first factor in the tensor-product space (see Ref.~\cite[Proposition~4.2]{KW20} for a proof of~\eqref{eq:choi-implement-map}).
\end{definition}

\begin{lemma}\label{lemma:choiprop}
Given a linear map $\Phi : L(\mathcal{H})\rightarrow  L(\mathcal{H}')$, its Choi operator $\Gamma^{\Phi}$ is a Hermitian operator if and only if $\Phi$ is a Hermiticity-preserving map.
\end{lemma}

\begin{proof}
First, suppose that $\Phi$ is Hermiticity preserving. From~\eqref{eq:choiop}, we can write
\begin{align}
(\Gamma^{\Phi})^{\dag}  & =\left(  \sum_{i,j=0}^{N-1}|i\rangle \langle j|\otimes
\Phi(|i\rangle \langle j|)\right)  ^{\dag}
 =\sum_{i,j=0}^{N-1}|i\rangle \langle j|^{\dag}\otimes\left(  \Phi(|i\rangle
 \langle j|)\right)  ^{\dag}\\
& \overset{\mathrm{(a)}}{=}  \sum_{i,j=0}^{N-1}|i\rangle \langle j|^{\dag}\otimes\Phi(|i\rangle \langle
j|^{\dag}) =\sum_{i,j=0}^{N-1}|j\rangle \langle i|\otimes\Phi(|j\rangle \langle i|) =\Gamma^{\Phi},
\end{align}
where equality (a) follows from Definition~\ref{def:hermilmap}. 

Now suppose that $\Gamma^{\Phi}$ is Hermitian. Then it follows from~\eqref{eq:choi-implement-map} that $\Phi(X)$ is Hermitian if $X$ is Hermitian, so that $\Phi$ is Hermiticity preserving.
\end{proof}

We now recall some definitions related to the Lipschitz continuity and smoothness of a function.

\begin{definition}[Lipschitz continuity]\label{eq:lipcongen}
A function $f$ : $\mathbb{R}^{n} \rightarrow \mathbb{R}^{m}$ is $L$-Lipschitz continuous if there exists $L > 0$, such that  $\left\Vert f(\bs{x}) - f(\bs{x}') \right\Vert \leq L \left\Vert \bs{x} - \bs{x}' \right\Vert$ for all $ \bs{x}, \bs{x}' \in  \mathbb{R}^{n}$. We say that $L$ is a Lipschitz constant of $f$.
\end{definition}

For a univariate function $f$, suppose that the absolute value of its first derivative on an interval $I$ is bounded from above by a positive real $L$, i.e., $\forall x \in I : \left|\sfrac{df(x)}{dx} \right| \leq L$.  Then $L$ is a Lipschitz constant of $f$. 

\begin{lemma}[Lipschitz constant for a multivariate function]\label{lemma:lipmul}
For a function $f$ : $\mathbb{R}^{n} \rightarrow \mathbb{R}$ with bounded partial derivatives, the value $L = \sqrt{n} \max_{i} \left \{\sup_{\bsx} \left| \partial f(\bsx)/\partial x_{i} \right| \right \}$ is a Lipschitz constant of $f$.
\end{lemma}

\begin{proof}
See Appendix~\ref{app:lipmul}. 
\end{proof}

\begin{lemma}[Lipschitz constant for a multivariate vector-valued function]\label{lemma:lipvec}
For a multivariate vector-valued function $f$ : $\mathbb{R}^{n} \rightarrow \mathbb{R}^{m}$, if each of its components, $f_{i}$, is $L_{i}$-Lipschitz, then $L = \sqrt{\sum_{i=1}^m L_{i}}$ is a Lipschitz constant of $f$.
\end{lemma}
\begin{proof}
See Appendix~\ref{app:lipvec}. 
\end{proof}

\begin{definition}[Smoothness]
A function $f$ : $\mathbb{R}^{n} \rightarrow \mathbb{R}^{m}$ is $\ell$-smooth if its gradient is $\ell$-Lipschitz, i.e., if there exists $\ell > 0$ such that $\left\Vert \nabla f(\bs{x}) - \nabla f(\bs{x}') \right\Vert \leq \ell \left\Vert \bs{x} - \bs{x}' \right\Vert$ for all $ \bs{x}, \bs{x}' \in  \mathbb{R}^{n}$.
\end{definition}

As most of the objective functions that we deal with in this paper are non-convex in some parameters, it is vital to focus on local optimality rather than global optimality because finding globally optimal points of a non-convex function is generally intractable. Therefore, the notion of $\epsilon$-stationary points is important for us. Intuitively, a point is $\epsilon$-stationary if the norm of the gradient at that point is very small. Formally, we define an $\epsilon$-stationary point as follows:
\begin{definition}[$\epsilon$-stationary point]\label{def:estationary}
Let $f : \mathbb{R}^{n} \rightarrow \mathbb{R}^{m}$ be a differentiable function, and let $\epsilon \geq 0$. A point $\bsx \in \mathbb{R}^{n}$ is an $\epsilon$-stationary point of  $f$ if $\lV \nabla f(\bsx) \rV \leq \epsilon $. \end{definition}

The above definition applies to first-order stationary points. This definition is important because we use inexact first-order solvers that converge to approximate first-order stationary points, and such a definition acts as a stopping criterion.

\begin{definition}[Polyak--Łojasiewicz (PL) Inequality]\label{def:PLineq}
A function $f : \mathbb{R}^{n} \rightarrow \mathbb{R}$ satisfies the PL inequality if, for some $\mu > 0 $, the following holds for all $\bsx \in \mathbb{R}^{n}$:
\begin{equation}
    \frac{1}{2} \lV f(\bsx)\rV^{2} \geq \mu (f(\bsx) - f^*),
\end{equation}
where $f^*$ is the globally optimal value of $f$. 
\end{definition}

In other words, the above inequality implies that every stationary point is a global minimum.

\subsection{Semidefinite Programming}

In this section, we recall some basic aspects of semidefinite programming~\cite{BV04,Wat18}. A semidefinite program is an optimization problem for which the goal is to optimize a linear function over the intersection of the positive semidefinite cone with an affine space. SDPs extend linear programs (LPs), such that the vector inequalities of LPs are generalized to matrix inequalities.

To begin with, recall the standard or canonical form of an SDP~\cite{BV04}:
\begin{equation}
\label{eq:sdpgeneralform}
\begin{aligned}
\sup_{X\succcurlyeq 0}&\ \ \operatorname{Tr}[CX] \\
\text{subject to}&\ \  \operatorname{Tr}[A_{i}X] \leq b_{i};  \ \ \forall i\in [M],
\end{aligned}
\end{equation}
where $C, A_{1}, \ldots, A_{M} \in \mathcal{S}^{N}$ and  $b_{1}, \ldots, b_{M} \in \mathbb{R}$. The standard form of an SDP is widely known for designing approximation algorithms for combinatorial optimization problems. A more general form of an SDP, as considered in~\cite{Wat18}, is as follows:
\begin{equation}
\label{eq:sdpgeneralformmapprimal}
\begin{aligned}
p^* \coloneqq \sup_{X\succcurlyeq 0}&\ \ \operatorname{Tr}[CX] \\
\text{subject to}&\ \  \Phi(X) \preccurlyeq B,
\end{aligned}
\end{equation}
where $C \in \mathcal{S}^{N}$, the map $\Phi$ is  Hermiticity preserving (see Definition~\ref{eq:defmap1}), $B \in \mathcal{S}^{M}$, and $p^*$ is the optimal value of the program~\eqref{eq:sdpgeneralformmapprimal}.
The aforementioned form is known as the primal form of an SDP, and the corresponding dual form is given as
\begin{equation}
    \label{eq:sdpgeneralformmapdual}
\begin{aligned}
d^* \coloneqq \inf_{Y\succcurlyeq 0}&\ \ \operatorname{Tr}[BY] \\
\text{subject to}&\ \  \Phi^\dagger(Y) \succcurlyeq C,
\end{aligned}
\end{equation}
where $Y \in S_{+}^{M}$, the map $\Phi^\dagger$ is the adjoint  of $\Phi$ (see Definition~\ref{def:adlmap}), and $d^*$ is the optimal value of the program~\eqref{eq:sdpgeneralformmapdual}.

The duality theorem of SDPs states that if both the primal and dual programs have feasible solutions, then the optimum of the primal program is bounded from above by the optimum of the dual program. Under a very mild condition that the primal has a feasible solution and the dual has a strictly feasible solution (or vice versa), strong duality holds; i.e., the duality gap (difference between $p^*$ and $d^*$) is closed. In the optimization literature, this condition is well known as Slater's condition (see Theorem~1.18 of~\cite{Wat18}).  Throughout this paper, we assume that strong duality holds.

\subsection{Variational Quantum Algorithms and Parameter-Shift Rule}\label{sec:vqa}

Variational quantum algorithms are hybrid quantum-classical algorithms, designed for solving optimization tasks with an objective function of the following form:
\begin{equation}\label{eq:vqacost}
    \mathcal{F}(\rho) = \sum_{k} g_{k}(\operatorname{Tr}[H_{k} \rho]),
\end{equation}
where $\rho \in \mathcal{D}^{N}$ is a density operator, $\{H_{k}\}_{k}$ is a set of problem-specific Hermitian operators, i.e., $H_{k} \in \mathcal{S}^{N}$ for all $k$, and $\{g_{k}\}_{k}$ is a problem-specific set of functions~\cite{VQA21,bharti2021noisy}. Additionally, each $g_{k}$ is a function of the expectation value of $H_{k}$ with respect to $\rho$.
The corresponding optimization problem is as follows:
\begin{equation}\label{eq:vqaop}
     \min_{\rho\in\mathcal{D}^{N}} \mathcal{F}(\rho).
\end{equation}

When the dimension $N$ is large, the evaluation of the expectation values of Hermitian operators with respect to $\rho$ is computationally intractable using classical algorithms. VQAs provide a quantum advantage because these hybrid algorithms attempt to circumvent this dimensionality problem by evaluating the expectation values of Hermitian operators using a quantum computer. Specifically, these algorithms utilize a parameterized quantum circuit to explore a problem-specific subspace of density operators and evaluate the expectation values with respect to these density operators. We discuss more about the quantum advantage of VQAs at the end of this section. First, let us discuss what we mean by a parameterized quantum circuit and how this circuit prepares a parameterized quantum state.\\

\noindent
\textbf{Parameterization:} Let $|\bsz\rangle_{RS}$ denote the all-zeros state of systems $R$ and $S$, each of which consists of $n$ qubits. Let $\rho_{S} \in \mathcal{D}^{2^n}$, and let $U_{RS}^{\rho}$ be a quantum circuit that prepares a purification $|\psi\rangle_{RS}$ of $\rho_{S}$ when $U_{RS}^{\rho}$ is applied to the initial state $|\bsz\rangle_{RS}$. Here, the subscripts~$S$  and $R$ are used to denote the system of interest and a reference system, respectively. VQAs simulate the space of density operators by parameterizing this quantum circuit as $U_{RS}(\bst)$, where $\bst = (\theta_{1}, \ldots, \theta_{r})^\top \in [0,2\pi]^{r}$. Specifically, a VQA applies a parameterized quantum circuit $U_{RS}(\bst)$ to the initial pure state $|\bsz\rangle_{RS}$ to generate a parameterized pure state 
\begin{equation}
|\psi(\bst)\rangle_{RS} \coloneqq U_{RS}(\bst)|\bsz\rangle_{RS}\in (\mathbb{C}^{2})^{\otimes 2n}.    
\end{equation}
Let $\rho(\bst)_S \coloneqq \operatorname{Tr}_R[|\psi(\bst)\rangle \langle \psi(\bst)|_{RS}]$ denote the reduced density operator of $|\psi(\bst)\rangle_{RS}$.
By using the fact that $\operatorname{Tr}[H_k \rho(\bst)]=\langle \psi(\bst) |_{RS}(I_{R} \otimes H_{k})|\psi(\bst)\rangle_{RS}$ and under the assumption that the parameterized circuit $U_{RS}(\bst)$ is fully expressive, the objective function in~\eqref{eq:vqacost} and its associated optimization problem~\eqref{eq:vqaop} can be written in terms of the parameter $\bst$ as follows:
\begin{align}
    \mathcal{F}(\bst) & \coloneqq \sum_{k} g_{k}(\langle \psi(\bst) |_{RS}(I_{R} \otimes H_{k})|\psi(\bst)\rangle_{RS}) \\
    & = \sum_{k} g_{k}(\langle I_{R} \otimes  H_{k}\rangle_{\bst}) , \label{eq:vqacostpure} \\
     \bst^*  & \coloneqq \argmin_{\bst \in [0, 2\pi]^{r}} \mathcal{F}(\bst)
     \label{eq:vqaoppure},
\end{align}
where it is implicit that $H_{k}$ acts on system $S$.

\begin{assumption}\label{as:objfunc}
We assume that the objective function in~\eqref{eq:vqacostpure} is `faithful,' which means that the minimum of this objective function corresponds to the optimal value of the problem~\eqref{eq:vqaop}.
\end{assumption}

A parameterized quantum circuit $U(\bst)$ is also known as a variational ansatz, and its choice plays an important role in obtaining an approximation of the optimal value. Throughout this paper, we use the terms ``parameterized quantum circuit'' and ``variational ansatz'' interchangeably. Furthermore, there are problem-specific ansatzes as well as problem-independent ansatzes. For our case, we use a problem-independent ansatz having the following form:
\begin{equation}\label{eq:vqauni}
    U(\bst) = U_{r}(\theta_{r})U_{r-1}(\theta_{r-1})\cdots U_{1}(\theta_{1}),
\end{equation}
where each unitary $U_{j}(\theta_{j})$ is written as
\begin{equation}
    U_{j}(\theta_{j}) = e^{-i \theta_{j} H_{j}} W_{j},
\end{equation}
with $W_{j}$ as an unparameterized unitary.
\begin{assumption}
We assume that the number of parameters, $r$, of a parameterized circuit is $O(\mathrm{poly}(n))$.
\end{assumption}
\noindent
The above assumption is natural in the context of variational quantum algorithms, as we only have access to quantum circuits with short depth.

Each Hermitian operator $H_{k}$ is arbitrary. In general, we can express a Hermitian operator as a weighted sum of tensor products of Pauli operators,  i.e.,
\begin{align}\label{eq:Hsumsigma}
    H = \sum_{i=1}^{p} w_{i}\ \sigma_{i,1}\otimes\cdots \otimes \sigma_{i,n},
\end{align}
where $w_i \in \mathbb{R}$, $\sigma_{i,j} \in \{I, \sigma_{x}, \sigma_{y}, \sigma_{z}\}$, and $\sigma_{x}$, $\sigma_{y}$, and $\sigma_{z}$ are the Pauli operators. From~\eqref{eq:Hsumsigma} we see that the expectation value of an arbitrary Hermitian operator is equal to a linear combination of the expectation values of each tensor product of Pauli operators. However, in general, this linear combination may contain many terms.
\begin{assumption}\label{as:hermiop}
We assume that the number of terms in~\eqref{eq:Hsumsigma} is polynomial in $n$; i.e., $p = O(\mathrm{poly}(n))$.
\end{assumption}

The above three assumptions are standard in the literature on variational quantum algorithms. Overall, VQAs use a quantum circuit with parameter $\bst$ to estimate the expectation value of a given Hermitian matrix, and they utilize a classical optimizer to solve the optimization problem $\min_{\bst \in [0, 2\pi]^{r}} \mathcal{F}(\bst)$. 
In each round, a VQA updates the parameters of the quantum circuit according to a classical optimization algorithm.
We can update these parameters according to gradient-free approaches that include Nelder--Mead~\cite{NM65}, Simultaneous Perturbation Stochastic
Approximation (SPSA)~\cite{Spa92}, and Particle Swarm Optimization~\cite{KE95}. 
The main drawbacks of such gradient-free methods include slower convergence and less robustness against noise. On the contrary, if the evaluation of the gradient of a given objective function is not computationally expensive, then first-order methods like gradient descent are more suitable. 

VQAs have an advantage as they do not use a classical method to evaluate gradients. Instead, they evaluate the partial derivatives of $\mathcal{F}(\bst)$ with respective to each parameter using the same quantum circuit but with shifted parameters. It is possible to evaluate the partial derivatives on a  quantum computer using the parameter-shift rule~\cite{Li2017,Mitarai2018,SBG19}. Formally, we state the parameter-shift rule for the evaluation of  the partial derivatives of $\langle H \rangle_{\bst}$ with respect to its $j^{\text{th}}$ component, i.e., $\theta_{j}$, as
\begin{equation}
\label{eq:parashift}
    \frac{\partial \langle H \rangle_{\bst} }{\partial \theta_{j}} = \frac{1}{2} \left ( \langle H \rangle_{\bst + (\pi/2) \hat{\bse}_{j}} - \langle H \rangle_{\bst - (\pi/2) \hat{\bse}_{j}} \right ),
\end{equation}
where $\hat{\bse}_{j}$ is a unit vector with $1$ as its $j^{\text{th}}$ element and $0$ otherwise.

From~\eqref{eq:parashift}, we note that for the exact evaluation of the partial derivatives of $\langle H \rangle_{\bst}$ at a parameter value $\bst$, we need to compute the expectation values of this operator at the shifted parameters exactly.
However, the exact evaluation of the expectation value of an operator requires an infinite number of measurements on a quantum circuit. 
In reality, we perform only a limited number of measurements and then take the average of those values. 
Therefore, it is important to have an unbiased estimator of the expectation value. We assume that this is the case for all of our algorithms.
Such an assumption is standard in the literature, as it acts as a good starting point for understanding the nature of the algorithms under such settings.  \\

\noindent
\textbf{Unbiased Estimator:} First, we consider a simple objective function consisting of a single expectation value term, i.e., 
\begin{equation}\label{eq:vqasimplecost}
  \mathcal{F}(\bst) = \langle H\rangle_{\bst}.
\end{equation}
Now, we define a $k$-sample mean unbiased estimator of this expectation value as follows:

\begin{definition}[$k$-sample mean unbiased estimator of expectation value]\label{def:expestimator}
Given a quantum circuit $U(\bst)$ with parameter $\bst \in [0, 2\pi]^{r}$, we define $u^{H}_{k}(\bst)$ as an average of $k$ measurements of the observable $H$ with respect to the pure state $U(\bst)| \bs{0}\rangle$. It is a $k$-sample mean unbiased estimator of~\eqref{eq:vqasimplecost} if the following holds:
\begin{equation}\label{eq:vqaexpestimator}
    \mathbb{E}[u^{H}_{k}(\bst)] = \langle H\rangle_{\bst}.
\end{equation}
\end{definition}

Next, we define an unbiased estimator of the partial derivatives of $\langle H \rangle_{\bst}$. According to the parameter-shift rule in~\eqref{eq:parashift}, the partial derivative of the expectation values of a Hermitian operator is a linear combination of its expectation values with shifted parameters. Setting
\begin{equation}\label{eq:k-sample-mean-est-gradient-def}
    g_{j}^{H, k}(\bst) \coloneqq \frac{1}{2} \left ( u^{H}_{k}(\bst + (\pi/2) \hat{\bse}_{j}) - u^{H}_{k}(\bst - (\pi/2) \hat{\bse}_{j} \right ),
\end{equation}
it follows that $g_{j}^{H, k}(\bst)$ is an unbiased estimator for that partial derivative because
\begin{align}\label{eq:vqapdestimator}
    \mathbb{E}\!\left[g_{j}^{H, k}(\bst)\right] & = \mathbb{E}\!\left[\frac{1}{2} \!\left ( u^{H}_{k}(\bst + (\pi/2) \hat{\bse}_{j}) - u^{H}_{k}(\bst - (\pi/2) \hat{\bse}_{j} \right )\right] \\ 
    & = \frac{1}{2} \left ( \mathbb{E}[u^{H}_{k}(\bst + (\pi/2) \hat{\bse}_{j})] - \mathbb{E}[u^{H}_{k}(\bst - (\pi/2) \hat{\bse}_{j})] \right ) \\ 
    & = \frac{1}{2} \left ( \langle H \rangle_{\bst + (\pi/2) \hat{\bse}_{j}} - \langle H \rangle_{\bst - (\pi/2) \hat{\bse}_{j}} \right ) \\ 
    & = \frac{\partial \langle H \rangle_{\bst} }{\partial \theta_{j}}.
\end{align}
We denote an unbiased estimator of the full gradient, i.e., $\nabla_{\theta} \langle H \rangle_{\bst}$, as $g^{H, k}(\bst)$. Furthermore, we evaluate the mean square error of $g^{H, k}(\bst)$ as follows,
\begin{align}\label{as:variance}
    & \!\!\!\!\!\! \mathbb{E}\!\left[ \lV g^{H, k}(\bst) - \nabla_{\theta} \langle H \rangle_{\bst} \rV^{2} \right] \notag \\
    & = \mathbb{E}\!\left [\sum_{j = 1}^{r} \left (g^{H, k}_{j}(\bst) - \frac{\partial \langle H \rangle_{\bst} }{\partial \theta_{j}} \right)^{2} \right] \\
    & = \sum_{j = 1}^{r} \mathbb{E}\!\left [ \left(g^{H, k}_{j}(\bst) - \frac{\partial \langle H \rangle_{\bst} }{\partial \theta_{j}}\right)^{2}\right]\\
    & = \sum_{j = 1}^{r} \text{Var}\!\left(g^{H, k}_{j}(\bst)\right)\\
    & = \sum_{j = 1}^{r} \frac{1}{4} \text{Var}\!\left ( u^{H}_{k}(\bst + (\pi/4) \hat{\bse}_{j}) - u^{H}_{k}(\bst - (\pi/4) \hat{\bse}_{j} \right )\\ 
    & =  \frac{1}{4} \sum_{j = 1}^{r} \text{Var}\!\left ( u^{H}_{k}(\bst + (\pi/4) \hat{\bse}_{j})\right) + \text{Var}\!\left (u^{H}_{k}(\bst - (\pi/4) \hat{\bse}_{j} \right )\\
    & =  \frac{1}{4} \sum_{j = 1}^{r} \frac{1}{k}\text{Var}\!\left ( u^{H}_{1}(\bst + (\pi/4) \hat{\bse}_{j})\right) + \frac{1}{k}\text{Var}\!\left ( u^{H}_{1}(\bst - (\pi/4) \hat{\bse}_{j})\right),
\end{align}
where we denote the variance of a random variable $X$ as $\text{Var}(X)$. The fourth equality follows from the definition of $g^{H, k}_{j}(\bst)$, given by~\eqref{eq:k-sample-mean-est-gradient-def}. The fifth equality uses the fact that $\text{Var}(X-Y) =\text{Var}(X) + \text{Var}(Y) $ if $X$ and $Y$ are independent random variables. For our case, the sample means $ u^{H}_{k}(\bst + (\pi/4) \hat{\bse}_{j})$ and $u^{H}_{k}(\bst - (\pi/4) \hat{\bse}_{j})$ are independent because they are computed using two different quantum circuit evaluations. For computing gradients in our algorithms, we take a sufficient number of samples, $k$, such that the mean square error is very small throughout the paper. This implies that the stochastic gradient ($g^{H, k}(\bst)$) is almost equal to the exact gradient ($\nabla_{\theta} \langle H \rangle_{\bst}$).

\textbf{Potential quantum advantage of a VQA when $N$ is large:} According to~\eqref{eq:Hsumsigma}, we consider a Hermitian operator $H$ consisting of a sum of $p$ weighted Pauli strings where $p = O(\text{poly}(n))$ (Assumption~\ref{as:hermiop}). 
We can measure  a given Pauli string $P_{i}$ of arbitrary size with respect to a given quantum state in constant time. If the desired precision of the expectation value's estimate  is $\epsilon$, then we need to make $O(1/ \epsilon^{2})$ repetitions of the procedure. Here the procedure consists of preparing the quantum state and then measuring the expectation value. As we have $p$ Pauli strings, overall we need no more than $O(\max_{i} \left| w_{i}\right|^{2} p/ \epsilon^{2})$ repetitions to estimate the expectation value of $H$ to precision $\epsilon$.
In contrast, if the evaluation of the expectation value of a Pauli string $P_{i}$ with respect to a general quantum state is conducted using known classical algorithms, it appears that we need $O(2^{n})$ time, as well as space. This is because the size of $P_{i}$ is $2^{n} \times 2^{n}$, and the same goes for the size of the matrix needed to represent the quantum state. Therefore, if all the above mentioned assumptions hold, then we have an exponential advantage over the best known classical algorithms for estimating the expectation value of a Hermitian operator. A similar argument can be provided for estimating gradients because the partial derivatives of the expectation value with respect to quantum circuit parameters depend on the expectation value evaluated on shifted parameters.

\section{Variational Quantum Algorithms for SDPs}\label{section:formulations}

\subsection{General Form (GF) of SDPs}

\label{sec:reform-gen-SDP}

In this section, we consider the general form of an SDP, as given in~\eqref{eq:sdpgeneralformmapprimal}. We write it concisely as follows:
\begin{equation}
p^* = \sup_{X\succcurlyeq0}\left \{  \operatorname{Tr}[CX]:\Phi(X)\preccurlyeq B \right\}.
\end{equation}
Next, we modify the above formulation as follows:
\begin{align}
p^* & = \sup_{X\succcurlyeq 0}\left\{  \operatorname{Tr}[CX]:\Phi(X)\preccurlyeq B\right\}
\\
&  \overset{\mathrm{(a)}}{=} \sup_{X\succcurlyeq0}\left\{  \operatorname{Tr}[CX]+\inf_{Y\succcurlyeq0}\left\{
\operatorname{Tr}[(B-\Phi(X))Y]\right\}  \right\}  \\ 
& \overset{\mathrm{(b)}}{=} \sup_{X\succcurlyeq0}\inf_{Y\succcurlyeq0}\left\{  \operatorname{Tr}[CX]+\operatorname{Tr}
[BY]-\operatorname{Tr}[Y\Phi(X)]\right\}  \\ 
&  \overset{\mathrm{(c)}}{=} \sup_{X\succcurlyeq0}\inf_{Y\succcurlyeq0}\left\{  \operatorname{Tr}[CX]+\operatorname{Tr}
[BY]-\operatorname{Tr}[(X^\top\otimes Y)\Gamma^{\Phi}]\right\}  \\ 
&  = \sup_{X\succcurlyeq0}\inf_{Y\succcurlyeq0}\left\{  \operatorname{Tr}[C^\top X^\top
]+\operatorname{Tr}[BY]-\operatorname{Tr}[(X^\top\otimes Y)\Gamma^{\Phi
}]\right\}  \\ 
&  = \sup_{X\succcurlyeq0}\inf_{Y\succcurlyeq0}\left\{  \operatorname{Tr}[C^\top
X]+\operatorname{Tr}[BY]-\operatorname{Tr}[(X\otimes Y)\Gamma^{\Phi}]\right\}
\\
&  \overset{\mathrm{(d)}}{=}  \sup_{\substack{\rho\in\mathcal{D}^{N},\\\lambda\geq0}}\inf_{\substack{\sigma\in\mathcal{D}^{M},\\\mu\geq
0}}\left\{  \lambda\operatorname{Tr}[C^\top\rho]+\mu\operatorname{Tr}
[B\sigma]-\lambda\mu\operatorname{Tr}[\Gamma^{\Phi
}(\rho\otimes\sigma)]\right\}.
\end{align}
Equality (a) follows due to the equivalence between both problems.  If we pick a primal PSD variable $X$ that does not satisfy the constraint, i.e., $\Phi(X) \npreccurlyeq B$, then the inner minimization results in the value $-\infty$. This is because, in such a case, there exists at least one negative eigenvalue of $B - \Phi(X)$. If we set $Y = s|e_{i}\rangle \langle e_{i} |$, where $|e_{i}\rangle$ is a unit vector in the negative eigenspace corresponding to that eigenvalue, then, due to the inner minimization, we can take the limit $s \rightarrow \infty$. This in turn implies that $\inf_{Y\succcurlyeq0}\left\{
\operatorname{Tr}[(B-\Phi(X))Y]\right\} = -\infty$. In other words, the inner minimization forces the outer maximization to pick a feasible $X$ and imposes an infinite penalty if chosen otherwise. 

Equality (b) follows by taking the infimum outside. This formulation is the Lagrangian of the original primal problem~\eqref{eq:genprimalcon}, where
\begin{equation}
    \label{eq:genprimalconlag}
    \Lagr(X,Y)\coloneqq  \operatorname{Tr}[CX]+\operatorname{Tr}
[BY]-\operatorname{Tr}[Y\Phi(X)]
\end{equation}
is a Lagrangian and $Y$ is a dual PSD variable.

Equality (c) follows from the definition of the Choi representation of the linear map~$\Phi$ (see Definition~\ref{def:choiop}), where $\Gamma^{\Phi}$ is the Choi operator of $\Phi$. According to Lemma~\ref{lemma:choiprop}, the Choi operator $\Gamma^{\Phi}$ is Hermitian.

Equality (d) follows from the substitution: $X = \lambda \rho$ and $Y = \mu \sigma$, where $ \rho \in \mathcal{D}^{N}, \sigma \in \mathcal{D}^{M}$, $\lambda = \operatorname{Tr}[X]$, and $\mu = \operatorname{Tr}[Y]$.

Now, we are interested in solving the following unconstrained optimization problem, i.e., the equality (d):
\begin{equation}\label{eq:genformop}
    p^{*} =   \sup_{\substack{\rho\in\mathcal{D}^{N},\\\lambda\geq0}}\inf_{\substack{\sigma\in\mathcal{D}^{M},\\\mu\geq
0}}\left\{  \lambda\operatorname{Tr}[C^\top\rho]+\mu\operatorname{Tr}
[B\sigma]-\lambda\mu\operatorname{Tr}[\Gamma^{\Phi
}(\rho\otimes\sigma)]\right\}.
\end{equation}
This optimization problem is now expressed in terms of the expectation values of the Hermitian operators $C^{\top},\ B,$ and $\Gamma^{\Phi}$ with respect to the density operators $\rho$, $\sigma$, and $\rho \otimes \sigma$, respectively.

\subsubsection{Variational Quantum Algorithm for SDPs in GF}

\label{sec:vqagf}

As stated earlier, when the dimension $N$ of the Hermitian operators is large, solving the problem~\eqref{eq:genformop} is generally intractable using a gradient-based classical algorithm. Therefore, we propose a variational quantum algorithm for the optimization problem~\eqref{eq:genformop}.

First, we introduce a parameterization of the density operators, i.e., $\rho$ and $\sigma$, by using parameterized quantum circuits. Second, we optimize the modified objective function of the problem~\eqref{eq:genformop} over the parameters of those quantum circuits using our variational quantum algorithm.\\

\noindent
\textbf{Parameterization:}
Let $\rho_{S_1}(\bst_{1})$ be the density operator prepared by first applying the quantum circuit $U_{R_1S_1}(\bst_{1})$ to the all-zeros state of the quantum system $R_1S_1$ and then tracing out the system $R_1$. Similarly, let $\sigma_{S_2}(\bst_{2})$ be the density operator prepared by first applying the quantum circuit $U_{R_2S_2}(\bst_{2})$ to the all-zeros state of the quantum system $R_2S_2$ and then tracing out the system $R_2$. Here, $\bst_{1} \in [0, 2\pi]^{r_{1}}$ and $\bst_{2} \in [0, 2\pi]^{r_{2}}$. Also, we set $r_{1}, r_{2} = O(\text{poly}(n))$. Defining
\begin{align}
    |\psi(\bst_1)\rangle_{R_1 S_1} & \coloneqq U_{R_1S_1}(\bst_{1})| \bs{0}\rangle_{R_1S_1}, \\
    |\varphi(\bst_2)\rangle_{R_2 S_2} & \coloneqq U_{R_2S_2}(\bst_{2})| \bs{0}\rangle_{R_2S_2},
\end{align}
we have that
\begin{align}
    \rho_{S_1}(\bst_{1}) & = \operatorname{Tr}_{R_1}[|\psi(\bst_1)\rangle\langle\psi(\bst_1)|_{R_1 S_1}], \\
    \sigma_{S_2}(\bst_{2}) & = \operatorname{Tr}_{R_2}[|\varphi(\bst_2)\rangle\langle\varphi(\bst_2)|_{R_2 S_2}].
\end{align}

\begin{figure}
\centering
\includegraphics[width=\textwidth]{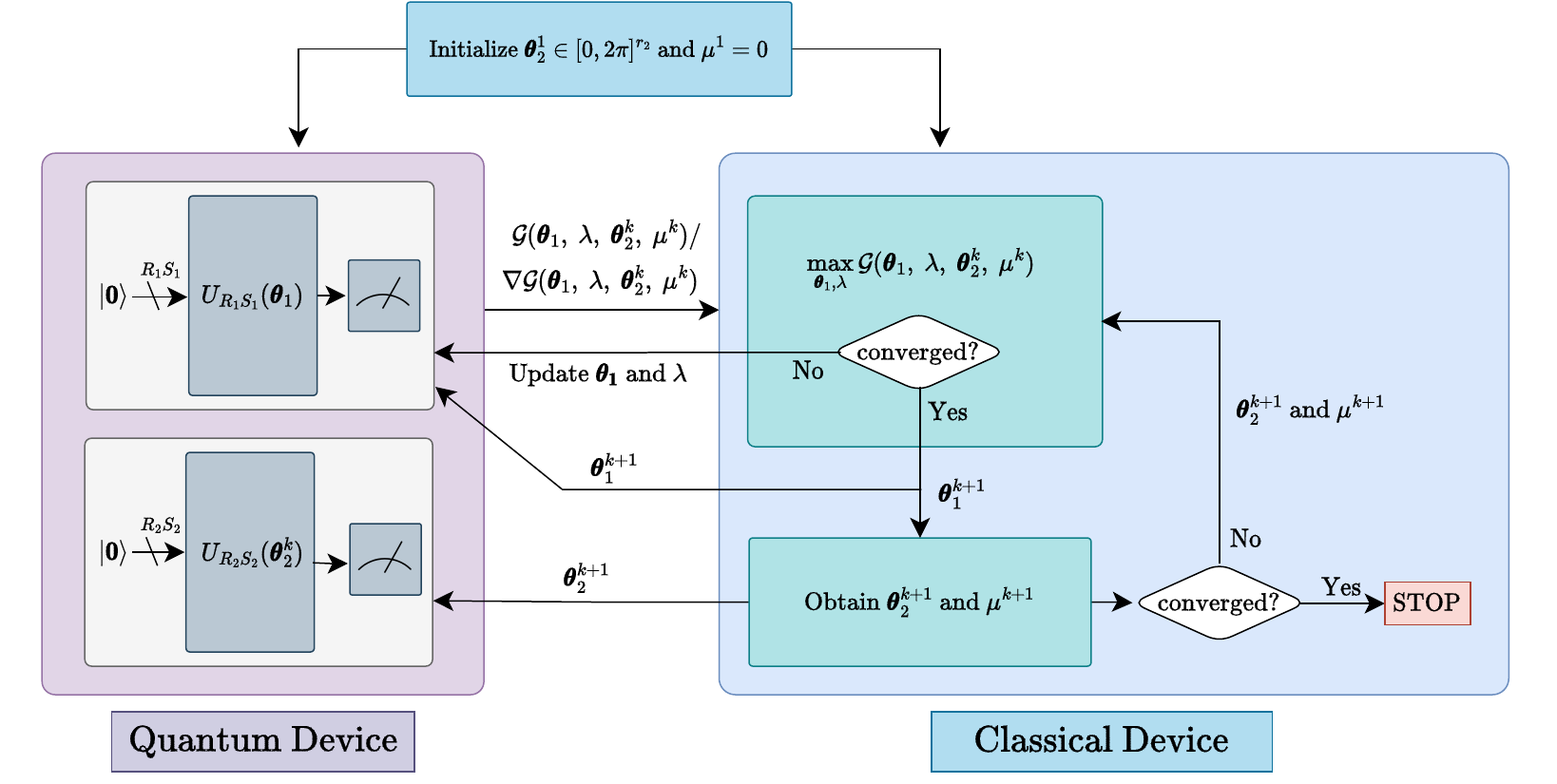}
\caption{This figure depicts iVQAGF algorithm where we utilize two parameterized quantum circuits, i.e., $U_{R_1S_1}(\bst_{1})$ and $U_{R_2S_2}(\bst_{2})$.}
\label{fig:iVQAGF}
\end{figure}

Furthermore, the parameterized quantum circuits, i.e., $U_{R_1S_1}(\bst_{1})$ and $U_{R_2S_2}(\bst_{2})$, are of the form shown in~\eqref{eq:vqauni}. As these quantum circuits generate purifications of density operators, the following equalities hold:
\begin{align}
    \operatorname{Tr}[C^{\top} \rho_{S_1}(\bst_{1})] & = \langle\psi(\bst_1)|_{R_1 S_1} \left (I_{R_1}\otimes C_{S_1}^{\top} \right)  |\psi(\bst_1)\rangle_{R_1 S_1} \\
    & = \langle I\otimes C^{\top} \rangle_{\bst_{1}},
    \\
    \operatorname{Tr}[B \sigma_{S_2}(\bst_{2})] & = \langle\varphi(\bst_2)|_{R_2 S_2} \left (I_{R_2}\otimes B_{S_2} \right )  |\varphi(\bst_2)\rangle_{R_2 S_2}  \\
    & = \langle I\otimes B \rangle_{\bst_{2}},\\
    \operatorname{Tr}[\Gamma^{\Phi}( \rho_{S_1}(\bst_{1}) \otimes \sigma_{S_2}(\bst_{2}))]  & = \langle\phi(\bst_{1},\bst_{2})|_{R_1R_2S_1S_2}   \left (I_{R_1}\otimes I_{R_2} \otimes \Gamma^{\Phi}_{S_1 S_2}\right) |\phi(\bst_{1},\bst_{2})\rangle_{R_1R_2S_1S_2}
    \notag \\
    & = \langle I\otimes I \otimes \Gamma^{\Phi}\rangle_{\bst_{1},\bst_{2}},
\end{align}
where
\begin{equation}
    |\phi(\bst_{1},\bst_{2})\rangle_{R_1R_2S_1S_2} \coloneqq |\psi(\bst_1)\rangle_{R_1 S_1} \otimes |\varphi(\bst_2)\rangle_{R_2 S_2}.
\end{equation}

Due to the parameterization of the density operators $\rho_{S_1}(\bst_{1})$ and  $\sigma_{S_2}(\bst_{2})$, the problem~\eqref{eq:genformop} transforms into the following optimization problem:
\begin{equation}
\label{eq:genfinalform-later}
    p^* \coloneqq \sup_{\substack{\bst_{1}  \in [0, 2\pi]^{r_{1}},\\\lambda\geq0}}\inf_{\substack{\bst_{2}  \in [0, 2\pi]^{r_{2}},\\\mu\geq
0}} \mathcal{G}(\bst_{1}, \lambda, \bst_{2}, \mu),
\end{equation}
where
\begin{equation}
\label{eq:genfinalformobj}
    \mathcal{G}(\bst_{1}, \lambda, \bst_{2}, \mu) \coloneqq  \lambda \langle  I \otimes C^{\top} \rangle_{\bst_{1}}+\mu\langle I \otimes B \rangle_{\bst_{2}}\\ -\lambda\mu\langle  I \otimes I \otimes \Gamma^{\Phi} \rangle_{\bst_{1},\bst_{2} }
\end{equation}
and we optimize over the space of quantum circuit parameters $\bst_{1}\in [0, 2\pi]^{r_{1}}$ and $\bst_{2}\in [0, 2\pi]^{r_{2}}$. 
As mentioned before, we assume that the objective function $ \mathcal{G}(\bst_{1}, \lambda, \bst_{2}, \mu)$ is faithful, which means that the global optimal value of the optimization problem~\eqref{eq:genfinalform-later} is equal to $p^*$.
Additionally, the objective function of~\eqref{eq:genfinalform-later} is generally non-convex as a function of the quantum circuit parameters $\bst_{1}$ and $\bst_{2}$. Hence, the max-min problem~\eqref{eq:genfinalform-later} is a non-convex--non-concave optimization problem. Due to the fact that obtaining a globally optimal point of a non-convex--non-concave function is generally NP-hard~\cite[Section 2.1]{Danilova2022}, we focus on finding first order $\epsilon$-stationary points of~\eqref{eq:genfinalform-later}. As a global optimal point is also a stationary point, if we use techniques to initialize quantum circuit parameters according to the problem at hand such that these parameters lie in the vicinity of a global optimal point, then our algorithm will converge to that point.

\begin{definition}[First order $\epsilon$-stationary point of~\eqref{eq:genfinalform-later}] A point $(\bst_{1}, \lambda, \bst_{2}, \mu)$ is a first order $\epsilon$-stationary point of~\eqref{eq:genfinalform-later} if and only if $ \lV \nabla \mathcal{G}(\bst_{1}, \lambda, \bst_{2}, \mu)\rV \leq \epsilon$.
\end{definition}

We propose a variational quantum algorithm for obtaining a first order $\epsilon$-stationary point of~\eqref{eq:genfinalform-later}. We call this algorithm \textit{inexact Variational Quantum Algorithm for General Form} (iVQAGF), and the pseudocode for this algorithm is provided in Algorithm~\ref{algo:iVQAGF}. It is an inexact version because we solve the subproblem involving a maximization (Step~4) to approximate stationary points instead of solving it until global optimality is reached, due to the non-convex nature of the objective function in terms of the quantum circuit parameters. Furthermore, iVQAGF is a hybrid quantum-classical algorithm, as we have two parameterized quantum circuits for estimating expectation values of Hermitian operators at any given step and a classical optimizer that updates the parameters of these quantum circuits as per the algorithm (see Figure~\ref{fig:iVQAGF}). 
 
\begin{algorithm}[ht]
\caption{\texttt{iVQAGF}($\Gamma^\Phi, C, B, \eta_{1}, \eta_{2}, \epsilon$)}\label{algo:iVQAGF}
\begin{algorithmic}[1]
\STATE \textbf{Input:} Hermitian operators $\Gamma^\Phi$, $C$, and $B$, learning rates $\eta_{1}, \eta_{2} > 0$, precision $\epsilon > 0$.
\vspace{0.7em}
\STATE \textbf{Initialization:} $\bs{\theta}^{1}_{2}\in [0, 2\pi]^{r_{2}}$, $\mu^{1} = 0$.\\
\vspace{0.7em}
\textcolor{gray}{\textit{\# For any step, expectation values of observables and their gradients are evaluated using parameterized quantum circuits.}}\\
\vspace{0.7em}
\FOR{$k =1, 2, \ldots,$}
\vspace{0.7em}
\STATE Maximize $\mathcal{G}(\cdot,\ \cdot,\ \bst_{2}^{k},\ \mu^{k})$ using a first order method such as gradient descent, where $\bst_{2}^{k},\ \mu^{k}$ are fixed, to obtain $(\bst_{1}^{k+1},\ \lambda^{k+1})$ such that the following holds:
\begin{equation}
     \lV \nabla \mathcal{G}(\bst_{1}^{k+1},\ \lambda^{k+1},\ \bst_{2}^{k},\ \mu^{k})\rV \leq \epsilon.\nonumber
\end{equation}
\STATE $\mu^{k+1} = \mu^{k} - \eta_{1} \nabla_{\mu} \mathcal{G}(\bst_{1}^{k+1},\ \lambda^{k+1},\ \bst_{2}^{k},\ \mu^{k})$
\vspace{0.7em}
\STATE $\bst_{2}^{k+1} = \bst_{2}^{k} - \eta_{2}\nabla_{\bst_{2}} \mathcal{G}(\bst_{1}^{k+1},\ \lambda^{k+1},\ \bst_{2}^{k},\ \mu^{k})$
\vspace{0.7em}
\IF{$ \lV \nabla \mathcal{G}(\bst_{1}^{k+1},\ \lambda^{k+1},\ \bst_{2}^{k+1},\ \mu^{k+1})\rV \leq \epsilon$}
\vspace{0.5em}
\STATE STOP and return $  \mathcal{G}(\bst_{1}^{k+1},\ \lambda^{k+1},\ \bst_{2}^{k+1},\ \mu^{k+1})$
\ENDIF
\ENDFOR
\end{algorithmic}
\end{algorithm}

At any step of the algorithm, the expectation value $\langle H \rangle_{\bst}$ of a Hermitian operator $H$ is evaluated using a quantum circuit with parameter $\bst$. Moreover, the partial derivatives of $\mathcal{G}(\bst_{1}, \lambda, \bst_{2}, \mu)$ with respect to the parameters $\bst_{1}$ and $\bst_{2}$ depend on the partial derivatives of the expectation values of the Hermitian operators $I\otimes C^{\top}$, $ I\otimes B$, and $I \otimes I \otimes \Gamma^{\Phi}$ with respect to those parameters. We evaluate the partial derivative of the expectation value of a Hermitian operator with respect to the quantum circuit parameters using the parameter-shift rule (see~\eqref{eq:parashift}). We do not explicitly mention these quantum-circuit calls in the main algorithm. 
 
\textbf{Unbiased Estimators:} Due to the fact that the evaluation of the expectation value of a Hermitian operator using a quantum circuit is  stochastic in nature, we use unbiased estimators of expectation values and their corresponding partial derivatives (see~\eqref{eq:vqaexpestimator} and~\eqref{eq:vqapdestimator}). Now, from~\eqref{eq:genfinalformobj}, we can see that $\nabla_{\bst_{1}} \mathcal{G}(\bst_{1}, \lambda, \bst_{2}, \mu)$ is a linear combination of $\nabla_{\bst_{1}} \langle  I \otimes C^{\top} \rangle_{\bst_{1}}$ and $\nabla_{\bst_{1}} \langle  I \otimes I \otimes \Gamma^{\Phi} \rangle_{\bst_{1},\bst_{2}}$. Therefore, the unbiased estimator of $\nabla_{\bst_{1}} \mathcal{G}(\bst_{1}, \lambda, \bst_{2}, \mu)$ is also a linear combination of the unbiased estimators of $\nabla_{\bst_{1}} \langle  I \otimes C^{\top} \rangle_{\bst_{1}}$ and $\nabla_{\bst_{1}} \langle  I \otimes I \otimes \Gamma^{\Phi} \rangle_{\bst_{1},\bst_{2}}$. Similarly, the unbiased estimator of $\nabla_{\bst_{2}} \mathcal{G}(\bst_{1}, \lambda, \bst_{2}, \mu)$ is a linear combination of the unbiased estimators of $\nabla_{\bst_{2}} \langle  I \otimes B \rangle_{\bst_{2}}$ and $\nabla_{\bst_{2}}\langle  I \otimes I \otimes \Gamma^{\Phi} \rangle_{\bst_{1},\bst_{2}}$. Also, the unbiased estimators of $\nabla_{\lambda} \mathcal{G}(\bst_{1}, \lambda, \bst_{2}, \mu)$ and $\nabla_{\mu} \mathcal{G}(\bst_{1}, \lambda, \bst_{2}, \mu)$ depend on a linear combination of the unbiased estimators of expectation values of Hermitian operators, i.e., $\langle I\otimes C^{\top}\rangle_{\bst_{1}},\ \langle I\otimes B\rangle_{\bst_{2}},$ and $\langle I \otimes I \otimes \Gamma^{\Phi}\rangle_{\bst_{1}, \bst_{2}}$. Finally, as we have the unbiased estimators of all these partial derivatives, we have an unbiased estimator of the full gradient $\nabla \mathcal{G}(\bst_{1}, \lambda, \bst_{2}, \mu)$ of the function $\mathcal{G}(\bst_{1}, \lambda, \bst_{2}, \mu)$.

\subsection{Equality Constrained Standard Form (ECSF) of SDPs}

\label{sec:ecsf}

In this section, we shift our focus to the standard form of SDPs. Due to their specific problem structure and the aforementioned assumption of it being weakly constrained (i.e, $N \gg M$), the design of more sophisticated variational quantum algorithms for solving them is possible. Moreover, we establish a convergence rate for one of the algorithms.

We consider the standard form of SDPs as given in~\eqref{eq:sdpgeneralform}. Here, the Hermiticity-preserving map $\Phi$ and the Hermitian operator $B$ have a diagonal form: 
\begin{align}
    \Phi(X) & = \text{diag}\left( \operatorname{Tr}[A_{1}X],\ldots, \operatorname{Tr}[A_{M}X] \right), \\
    B & = \text{diag}\left (b_{1},\ldots,b_{M} \right),
\end{align}
where $A_{1},\ldots,A_{M}\in\mathcal{S}^{N}$ and $b_{1}, \ldots, b_{M}\in \mathbb{R}$.

The equality constrained primal SDP is given as follows:
\begin{equation}\label{eq:eqprimal}
p^* = \sup_{X\succcurlyeq0}\left\{  \operatorname{Tr}[CX]:\bs{\Phi}(X) = \bs{b}\right\},
\end{equation}
where $\bs{b} = \left( b_{1},\ldots,b_{M}\right)^\top$ and $\bs{\Phi}$ is the vector form of the original linear map, i.e., 
$\bs{\Phi}(X) = \left ( \operatorname{Tr}[A_{1}X],\dots,\operatorname{Tr}[A_{M}X] \right)^\top$. Taking this formulation into account, we write the optimization over the Lagrangian $\Lagr(X, Y)$ in~\eqref{eq:genprimalconlag} as
\begin{align}\label{eq:eqprimallag}
    p^* = \sup_{X\succcurlyeq 0}\inf_{\bs{y}\in \mbbR^{M}} \Lagr(X, \bs{y}) ,
\end{align}
where
\begin{equation}
    \Lagr(X, \bs{y}) \coloneqq \operatorname{Tr}[CX] + \bs{y}^\top \left(\bs{b} - \bs{\Phi}(X)\right)
\end{equation}
and $\bs{y}\in \mbbR^{M}$ is the dual vector of the Lagrangian $\Lagr(X, \bs{y})$. As discussed earlier, the inner minimization with respect to $\bs{y}$ results in $-\infty$ if  $X$  violates the constraint in~\eqref{eq:eqprimal}. On the contrary, if there exists a PSD operator $X$ that satisfies the constraints, then~\eqref{eq:eqprimallag} reduces to $\sup_{X\succcurlyeq 0}\operatorname{Tr}[CX]$. This is because $\inf_{\bs{y}\in \mbbR^{M}}\bs{y}^\top \left(\bs{b} - \bs{\Phi}(X)\right) = 0$ in this case. Hence, there is an equivalence between~\eqref{eq:eqprimal} and~\eqref{eq:eqprimallag}, as indicated previously.

For the equality constrained problem~\eqref{eq:eqprimal}, we consider the Augmented Lagrangian $\Lagr_{c}(X, \bs{y})$ as the objective function, which consists of a quadratic penalty term $\frac{c}{2} \left \Vert \bs{b} - \bs{\Phi}(X)\right \Vert^{2}$ in addition to the terms of the original Lagrangian $\Lagr(X, \bs{y})$. Therefore, the optimization problem with the modified objective function can be written as follows:
\begin{equation}\label{eq:primalauglag}
    p^*= \sup_{X\succcurlyeq 0} \inf_{\bs{y}\in \mbbR^{M}} \Lagr_{c}(X, \bs{y}),
\end{equation}
where $c > 0$ is the penalty parameter and 
\begin{equation}
    \Lagr_{c}(X, \bs{y}) \coloneqq \operatorname{Tr}[CX] + \bs{y}^\top \left(\bs{b} - \bs{\Phi}(X)\right) - \frac{c}{2} \left \Vert \bs{b} - \bs{\Phi}(X)\right \Vert^{2}.
\end{equation}

The method for optimization associated with the Augmented Lagrangian formulation is known as the  Augmented Lagrangian Method (ALM),  independently introduced in~\cite{Hes69} and~\cite{Pow69}. This classical method involves the following three steps that iterate until convergence:
\begin{enumerate}
    \item $X^{k+1} \coloneqq \argmax_{X\succcurlyeq0} \Lagr_{c}(X, \bs{y}^{k})$,
    \item Update the dual variable $\bsy$ according to $\bs{y}^{k+1} \coloneqq \bs{y}^{k} - c \left(\bs{b} - \bs{\Phi}(X) \right)$,
    \item Update the penalty parameter $c$ according to  $c_{k+1} \coloneqq c_{0}\mu^{k+1}$, where $\mu > 1$.
\end{enumerate}
Due to the equivalence between~\eqref{eq:eqprimal} and~\eqref{eq:primalauglag},  iterating through the above steps until convergence leads to the primal optimal value. Hence, we can write
\begin{align}\label{eq:auglageqX}
  p^* & = \sup_{X\succcurlyeq 0}\inf_{\bs{y}\in \mbbR^{M}}\Lagr_{c}(X, \bs{y}) \\
  & = \sup_{X\succcurlyeq 0} \inf_{\bs{y}\in \mbbR^{M}} \left\{\operatorname{Tr}[CX] + \bs{y}^\top \left(\bs{b} - \bs{\Phi}(X)\right) - \frac{c}{2} \left \Vert \bs{b} - \bs{\Phi}(X)\right \Vert^{2}\right\}.  
\end{align}
ALM is a well-studied method for solving unconstrained optimization problems with convex objective functions. It is considered superior to methods that exclusively involve the original Lagrangian or exclusively involve the penalty term~\cite{Ber76}. ALM converges faster than the original Lagrangian method, as it involves a quadratic penalty term in addition to a linear penalty term of the original Lagrangian. Furthermore, it solves many issues associated with both these methods, such as ill-conditioning arising due to large values of the penalty parameter.

Now, we make a mild assumption that one of the constraints is a trace constraint on the primal PSD operator $X$. Specifically, let $A_{M} = \mathbb{I}$ and $b_{M} = \lambda$. Hence, the constraint is $\operatorname{Tr}[X] = \lambda$. Subsequently, substituting $X = \lambda \rho$, where $\rho \in \mathcal{D}^{N}$ in~\eqref{eq:auglageqX}, we obtain the following:
\begin{align}\label{eq:auglageqrho}
  p^* & = \sup_{\rho \in \mathcal{D}^{N}}\inf_{\bs{y}\in \mbbR^{M}}\Lagr_{c}(\lambda \rho, \bs{y}) \\
  & = \sup_{\rho \in \mathcal{D}^{N}} \inf_{\bs{y}\in \mbbR^{M}} \left\{\lambda \operatorname{Tr}[C\rho] + \bs{y}^\top \left(\bs{b} - \lambda \bs{\Phi}(\rho)\right) - \frac{c}{2} \left \Vert \bs{b} - \lambda \bs{\Phi}(\rho)\right \Vert^{2}\right\}.  
\end{align}
This optimization problem is now expressed in terms of the expectation values of the Hermitian operators $C,\ A_{1}, \ldots, A_{M}$ with respect to the density operator $\rho$. According to Assumption~\ref{as:hermiop}, the matrices $C,\ A_{1}, \ldots, A_{M}$ consist of $\text{poly}(n)$ terms.

\subsubsection{Variational Quantum Algorithm for SDPs in ECSF}

\label{sec:vqa-ecsf}

Solving the unconstrained optimization problem~\eqref{eq:auglageqrho} using ALM is computationally intractable if the dimension~$N$ of the operators is large. Therefore, we propose a variational quantum algorithm to solve this problem.

As before, we first introduce a parameterization of the density matrix $\rho$ and then optimize the modified objective function over the subspace of those parameters using our variational quantum algorithm. \\ 

\noindent
\textbf{Parameterization:} 
Let $U_{RS}(\bst)$ be a parameterized quantum circuit with parameter $\bst \in [0, 2\pi]^{r}$, and suppose that it acts on the all-zeros state of the quantum system $RS$ and prepares a purification of the density operator $\rho_{S}(\bst)$. The form of $U_{RS}(\bst)$ is given by~\eqref{eq:vqauni}. Moreover, the following equality holds because $U_{RS}(\bst)$ generates a purification of~$\rho_{S}(\bst)$:
\begin{equation}\label{eq:eqexp-3}
    \operatorname{Tr}[H \rho_{S}(\bst)] = \langle \bs{0} |_{RS} U_{RS}^{\dagger}(\bst) \left (I_R \otimes H_S \right)  U_{RS}(\bst) | \bs{0}\rangle_{RS} = \langle I\otimes H \rangle_{\bs{\theta}},
\end{equation}
where $H \in \{C, A_{1},\ldots,A_{M}\}$.

Due to the parameterization of the density operator $\rho_{S}(\bst)$, the problem~\eqref{eq:auglageqrho} transforms into the following optimization problem:
\begin{equation}\label{eq:eq-cons-final-opt}
  p^*  \coloneqq \sup_{\bs{\theta} \in [0, 2\pi]^{r}}\inf_{\bs{y}\in \mbbR^{M}}\Lagr_{c}(\bs{\theta}, \bs{y}),  
\end{equation}
where
\begin{align}
    \Lagr_{c}(\bs{\theta}, \bs{y}) & \coloneqq  \lambda \left \langle I \otimes C \right \rangle_{\bs{\theta}} + \bs{y}^\top \left(\bs{b} - \lambda \bs{\Phi}(\bs{\theta})\right) - \frac{c}{2} \left \Vert \bs{b} - \lambda \bs{\Phi}(\bs{\theta})\right \Vert^{2} , \label{eq:aug-lagr-def}\\
    \bs{\Phi}(\bst) & \coloneqq \left(\langle I \otimes A_{1} \rangle_{\bs{\theta}},\ldots,\langle I \otimes A_{M} \rangle_{\bs{\theta}} \right)^\top .
    \label{eq:aug-lagr-def-2}
\end{align}
As mentioned before, we assume that the objective function $ \Lagr_{c}(\bs{\theta}, \bs{y})$ is faithful, which means that the global optimal value of the optimization problem~\eqref{eq:auglageqpara} is equal to $p^*$.

\begin{figure}
\centering
\includegraphics[width=\textwidth]{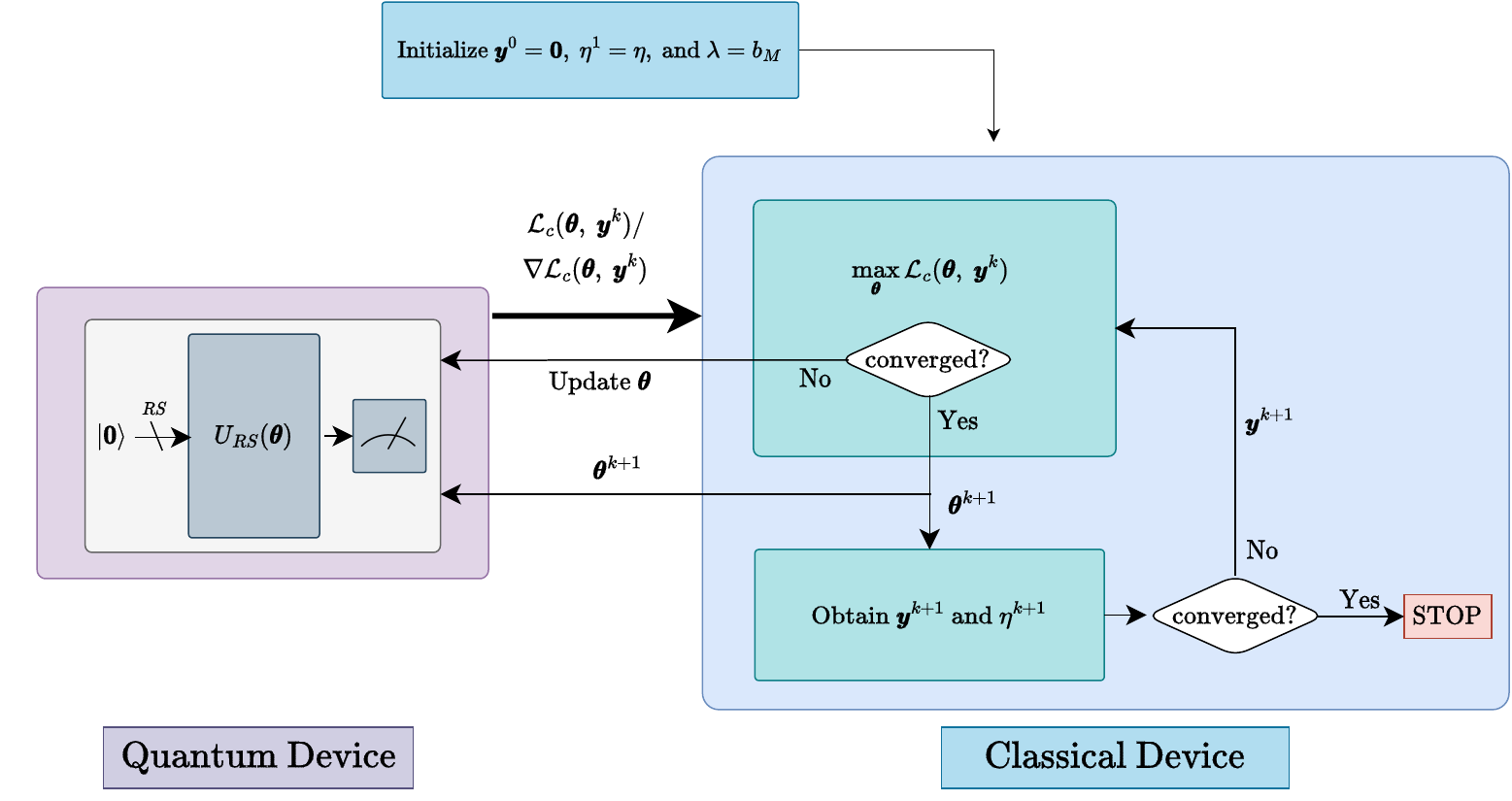}
\caption{This figure depicts the iVQAEC algorithm in which we utilize one parameterized quantum circuit, i.e., $U_{RS}(\bst)$.}
\label{fig:iVQAEC}
\end{figure}

We now focus on solving the optimization problem~\eqref{eq:auglageqpara} using the modified Lagrangian $\Lagr_{c}(\bs{\theta}, \bs{y})$. Note that, in general, the objective function provided above has a non-convex landscape with respect to the quantum circuit parameter $\bs{\theta}$. Also, it is linear in terms of the dual variable $\bs{y}$. Therefore, the optimization problem~\eqref{eq:auglageqpara} is a non-convex--concave optimization problem. For a convex-concave setting, the natural algorithm to consider is to first perform the maximization until a globally optimal point is reached and then update the dual variable, and repeat these steps until convergence (similar to ALM). 
On the contrary, for a non-convex--concave setting, solving the maximization until global optimality is reached is generally NP-hard~\cite[Section 2.1]{Danilova2022}. Therefore, we focus on obtaining approximate stationary points of~\eqref{eq:auglageqpara}. As a global optimal point is also a stationary point, if we use techniques to initialize quantum circuit parameters according to the problem at hand, such that these parameters lie in the vicinity of a global optimal point, then our algorithm will converge to that point.

\begin{definition}[First order $(\epsilon, c)$-stationary point of~\eqref{eq:auglageqpara}] For $\epsilon,c>0$, a point $\bs{\theta} \in [0, 2\pi]^{r}$ is a first order $(\epsilon,c)$-stationary point of~\eqref{eq:auglageqpara} if there exists $\bs{y} \in \mbbR^{M} $ such that the following hold: 
\begin{equation}
    \left \Vert \nabla_{\bs{\theta}}\Lagr_{c}(\bs{\theta}, \bs{y}) \right \Vert \leq \epsilon, \qquad \left \Vert \lambda \bs{\Phi}(\bs{\theta}) - \bs{b}\right \Vert \leq \epsilon.
\end{equation}
\end{definition}

The above condition acts as a measure of closeness to the first order $(\epsilon, c)$-stationary points of~\eqref{eq:auglageqpara}. Such a measure is useful to give a stopping criterion for the algorithm.

We now propose a variational quantum algorithm for obtaining a first order $(\epsilon, c)$-stationary point of~\eqref{eq:auglageqpara}. We call this algorithm \textit{inexact Variational Quantum Algorithm for Equality Constrained standard form} (iVQAEC). It is an inexact version because we solve the subproblem involving the maximization to approximate stationary points instead of solving it until global optimality is reached, due to the nonconcave nature of the objective function in terms of quantum circuit parameters. The pseudocode of iVQAEC is given by Algorithm~\ref{algo:iVQAEC}. 

\begin{algorithm}
\begin{doublespace}
\begin{algorithmic}[1]\onehalfspacing
\STATE \textbf{Input:} Hermitian operators $(C, A_{1}, \ldots, A_{M})$, vector  $\bs{b}$, final precision $\epsilon' > 0$, learning rate $\eta > 0$, and a constant $\mu > 1$.  \vspace{0.7em}
\STATE \textbf{Initialization:} $ \bs{y}^{1} = \bs{0}, \eta^1 = \eta, \lambda = b_{M}$\\
\vspace{0.7em}
\textcolor{gray}{\textit{\# For any step, expectation values of observables and their gradients are evaluated using a parameterized quantum circuit.}}\\
\vspace{0.7em}
\FOR{$k = 1, 2, \ldots,$}
\vspace{0.7em}
\STATE $c_{k} = -\mu^{k}$, $\epsilon_{k+1} = 1/\mu^{k}$
\vspace{0.7em}
\STATE Maximize $\mathcal{L}_{c_k}(\cdot, \bs{y}^{k})$, where $\bs{y}^{k}$ is kept fixed, to obtain $\bst^{k+1}$ such that the following holds:
\begin{equation}
     \lV \nabla \mathcal{L}_{c_k}(\bst^{k+1}, \bs{y}^{k})\rV \leq \epsilon_{k+1}\nonumber
\end{equation}
\vspace{0.7em}
\STATE$\eta^{k+1}  = \eta^{1} \min\Big\{ 
\frac{\left \Vert \lambda \bs{\Phi}( \bs{\theta}^{1}) - \bs{b}\right \Vert \ln^2 2 }{\left \Vert \lambda \bs{\Phi}(\bs{\theta}^{k+1}) - \bs{b}\right \Vert(k+1)\ln^2(k+2)},\ 1
\Big\}$ 
\vspace{0.7em}
\STATE$\bs{y}^{k+1}= \bs{y}^{k} -  \eta^{k+1}\left ( \bs{b} - \lambda \bs{\Phi}(\bs{\theta}^{k+1}) \right)$
\vspace{0.7em}
\IF{$\left \Vert \nabla_{\bs{\theta}}\Lagr_{c_k}(\bs{\theta}^{k+1}, \bs{y}^{k+1}) \right \Vert + \left \Vert \lambda \bs{\Phi}(\bs{\theta}^{k+1}) - \bs{b}\right \Vert  \leq \epsilon'$ }
\vspace{0.7em}
\STATE STOP and return $\Lagr_{c_k}(\bs{\theta}^{k+1}, \bs{y}^{k+1})$
\ENDIF
\ENDFOR
\end{algorithmic}
\caption{\texttt{iVQAEC}($C, \{A_{i}\}_{i=1}^{M},\ \bs{b}$, $\epsilon',\ \eta, \ \mu$)}
\label{algo:iVQAEC}
\end{doublespace}
\end{algorithm}

We run iVQAEC  on a classical computer, and at any step of the algorithm, the expectation value of a Hermitian operator $H \in \{C, A_{1}, \ldots, A_{M}\}$ (i.e., $\langle H \rangle_{\bst}$) is evaluated using a quantum circuit with parameter $\bst$ (see Figure~\ref{fig:iVQAEC}). Moreover, the partial derivative of $\mathcal{L}_{c}(\bst, \bsy)$ with respect to the parameter $\bst$ depends on the partial derivatives of the expectation values of the Hermitian operators $I\otimes C$, $I\otimes A_{1}$, \ldots, $I \otimes A_{M}$ with respect to $\bst$. For evaluating the partial derivatives of the latter with respect to $\bst$, we use the parameter-shift rule (see~\eqref{eq:parashift}). Moreover, we do not explicitly mention these calls to quantum circuits in the pseudocode of the algorithm. 
 
\textbf{Unbiased Estimators:} As mentioned earlier in the previous section, due to the fact that evaluation of the expectation value of a Hermitian operator using a quantum circuit is stochastic in nature, we use unbiased estimators of expectation values and their corresponding partial derivatives (see~\eqref{eq:vqaexpestimator} and~\eqref{eq:vqapdestimator}). Furthermore, it suffices to have \textit{independent} unbiased estimators of $\nabla_{\bst} \langle  I\otimes C \rangle_{\bst}$, $\nabla_{\bst} \langle  I\otimes A_{1} \rangle_{\bst}$, \ldots, $\nabla_{\bst} \langle  I\otimes A_{M} \rangle_{\bst}$, $ \langle  I\otimes A_{1} \rangle_{\bst}$, \ldots, $\langle  I\otimes A_{M} \rangle_{\bst}$ in order to construct an unbiased estimator of $\nabla_{\bst} \mathcal{L}_{c}(\bst, \bsy)$.

\subsubsection{Convergence Rate}

\label{sec:conanaleq}

In this section, we provide the convergence rate and total iteration complexity of iVQAEC in terms of the number of iterations of the for-loop at Step 3 of Algorithm~\ref{algo:iVQAEC} needed to arrive at an approximate stationary point of~\eqref{eq:eq-cons-final-opt}.

Recently, the authors of \cite{SEAGC19} studied the following non-convex optimization problem:
\begin{equation}\label{eq:opsahin}
    \inf_{x \in \mbbR^{d}}\left \{  f(x) + g(x) : \bs{A}(x) = 0\right\},
\end{equation}
where $f: \mbbR^{d} \rightarrow \mbbR$ is a continuously-differentiable non-convex function, $\bs{A} : \mbbR^{d} \rightarrow \mbbR^{m}$ is a nonlinear operator, and $g :  \mbbR^{d} \rightarrow \mbbR$ is a convex function. The constrained optimization problem~\eqref{eq:opsahin} is equivalent to the following unconstrained minimax formulation: 
\begin{equation}\label{eq:auglagsahin}
    \inf_{\bs{x} \in \mathbb{R}^{d}} \sup_{\bs{y} \in \mathbb{R}^{m}} \left \{ \Lagr_{\beta}(x, \bs{y})   + g(\bs{x}) \right \},
\end{equation}
where
\begin{equation}
    \Lagr_{\beta}(x, \bs{y}) \coloneqq  f(\bs{x}) + \bs{y}^\top A(\bs{x}) + \frac{\beta}{2} \left \Vert \bs{A}(\bs{x})\right \Vert^{2},
\end{equation}
$\beta > 1$, and $\Lagr_{\beta}(x, y)$ is the augmented Lagrangian. As the classic ALM algorithm is suited for solving optimization problems involving convex objective functions, the authors proposed an inexact Augmented Lagrangian method, known as the iALM algorithm for obtaining an approximate first order stationary point of~\eqref{eq:auglagsahin}. It is an inexact version because the subproblem involving a maximization is over a non-convex function. They claim that if we use an inexact solver for that subproblem at each iteration, then their algorithm converges to an approximate first order stationary point (see Theorem~4.1 of~\cite{SEAGC19}). They also provide the total iteration complexity of iALM in terms of the number of iterations of the for-loop (see Corollary~4.2 of~\cite{SEAGC19}).

Our problem~\eqref{eq:eq-cons-final-opt} also consists of an augmented Lagrangian term, which is non-convex in terms of the quantum circuit parameters $\bst$, and $g(\bst) = 0$. Additionally, for solving our problem, we use an inexact solver for the subproblem (Step 5) of iVQAEC, and the update rule for modifying the learning rate $\eta$ at each iteration (Step 6) is similar to their iALM algorithm.

Now, the following theorem characterizes iVQAEC's convergence rate and total iteration complexity for finding an approximate stationary point of the problem~\eqref{eq:eq-cons-final-opt} in terms of the number of iterations of the for-loop (Step 3).  The proof of this theorem is provided by the proofs of Theorem 4.1 and Corollary 4.2 of~\cite{SEAGC19}. 

\begin{theorem}[Convergence rate and total iteration complexity of iVQAEC]\label{thm:convrateiVQAML}
For integers $ k_{1} \geq k_{0} \geq 2$, consider the interval $K \coloneqq \{k_{0},\ldots,  k_{1}\}$, and let $ \{\bst^{k}\}_{k\in K}$ be the output sequence of iVQAEC  on the interval $K$, where $\bst^{k} \in [0, 2\pi]^r$ for all $k \in K$. Suppose that $f(\bst) \coloneq \langle I \otimes C \rangle_{\bst}$ and $\bs{A}(\bst) \coloneq \lambda \bs{\Phi}(\bst) - \bs{b}$  have Lipschitz constants $L_{f}$ and $L_{\bs{A}}$. Also suppose that the function $\Lagr_{c}(\cdot, \bs{y})$ is smooth for every $\bs{y} \in \mathbb{R}^{M}$ (i.e., its gradient is $L_{c,\bsy}$-Lipschitz continuous).
Additionally, suppose that there exists $\nu > 0$ such that
\begin{align}\label{eq:regularitysahin}
\nu \Vert \bs{A}(\bst^k)\Vert  
& \le \left \Vert -J_{\bs{A}}(\bst^k)^\top \bs{A}(\bst^k) \right \Vert,
\end{align}
for every $k\in K$. If an inexact solver is used in Step 5 of the algorithm, then $\bst^k$ is a first order $(\epsilon_{k}, c_k)$-stationary point of~\eqref{eq:eq-cons-final-opt} with 
\begin{align}
\epsilon_{k} & =  \frac{1}{c_{k-1}} \left(\frac{2(L_f + y_{\max} L_{\bs{A}} ) (1+ \eta^k L_{\bs{A}})}{\nu}+1\right) =: \frac{Q(f, \bs{A},\eta^1)}{c_{k-1}}, 
\label{eq:stat_prec_first}
\end{align}
for every $k\in K$, where $y_{\max}(\bst^1,\bs{y}^0,\eta^1)$ is given as
\begin{equation}
   y_{\max}(\bst^1,\bs{y}^0,\eta^1) = \Vert \bs{y}^{0} \Vert + \eta^{1}\Vert \bs{A}(\bst^{1})\Vert \ln^2 2.
\end{equation}
Additionally, if we use the Accelerated Proximal Gradient Method (APGM) from~\cite{HZ15} as an inexact solver for the subproblem (Step 5), then the algorithm finds a first-order $(\epsilon_T,c_T)$-stationary point, after $T$ calls to the solver, where
\begin{equation}
T = \mathcal{O}\!\left( \frac{r Q^3 }{\epsilon^{4}_{T}}\log_\mu\!\left( \frac{Q}{\epsilon_{T}} \right) \right) =   \tilde{\mathcal{O}}\!\left( \frac{r Q^{3} }{\epsilon^{4}_{T}} \right),
\end{equation}
with $Q \equiv Q(f, \bs{A},\eta^1)$, and $\mu$ is a constant as mentioned in Algorithm~\ref{algo:iVQAEC}.
\end{theorem}

The inequality given by~\eqref{eq:regularitysahin} is known as the PL inequality for minimizing $\left \Vert \bs{A}(\bst) \right \Vert^{2}$~\cite{KNS16}. This regularity condition is used in the proof of Theorem 4.1 of~\cite{SEAGC19} for obtaining a bound on $\left \Vert \bs{A}(\bst^{k}) \right \Vert$ for all $k \in K$. For our purposes, we assume that this condition holds for all $\bst \in [0, 2\pi]^r$. Then, Theorem~\ref{thm:convrateiVQAML} is a direct consequence of the following statements:
\begin{enumerate}[(A)]
    \item \textbf{Smoothness of $\Lagr_{c}(\cdot, \bs{y})$}: The function $\Lagr_{c}(\cdot, \bs{y})$ is smooth for every $\bs{y} \in \mbbR^{M}$ (i.e., its gradient is $L_{c,\bsy}$-Lipschitz continuous). \label{state:smoothnessL}
    \item \textbf{Lipschitz continuity of $f(\bst)$ and $\bs{A}(\bst)$}: The functions $f(\bst)$ and $\bs{A}(\bst)$ are Lipschitz continuous with constants $L_{f}$ and $L_{\bs{A}}$, respectively.\label{state:lipconCPhi}
\end{enumerate}
We prove the statements~\ref{state:smoothnessL} and \ref{state:lipconCPhi}
in Lemmas~\ref{lemma:smoothnessL} and \ref{lemma:lipconCPhi}, respectively.
Before we prove these statements, we state the following lemma that concerns Lipschitz continuity of the function $h : [0, 2\pi]^{r} \rightarrow \mathbb{R}$ defined as 
\begin{equation}
  h(\bst) \coloneqq \langle I \otimes O \rangle_{\bst} =   \langle \bs{0} | U^{\dagger}(\bst)(I \otimes O) U(\bst)|\bs{0}\rangle),  
\end{equation}
and its gradient $\nabla h(\bst)$, where $O \in \mathcal{S}^{N}$. 

\begin{lemma}[Lipschitz continuity of $h(\bst)$ and $\nabla h(\bst)$] \label{lemma:lipconexp} The function $h : [0, 2\pi]^{r} \rightarrow \mathbb{R}$ and its gradient, i.e., the vector-valued function $\nabla h : [0, 2\pi]^{r} \rightarrow \mathbb{R}^{r}$, are $L_{h}$-Lipschitz and $L_{\nabla h}$-Lipschitz continuous, respectively, for some $L_{h},L_{\nabla h} > 0$.
\end{lemma}

\begin{proof}
The proof is given in Appendix~\ref{app:lipconexp}.
\end{proof}

Next, we formally write the statements~\ref{state:smoothnessL} and \ref{state:lipconCPhi}
as lemmas.

\begin{lemma}[Smoothness of $\Lagr_{c}(\cdot, \bs{y})$] \label{lemma:smoothnessL}
For all $\bs{y} \in \mbbR^{M}$, there exists $L_{c, \bs{y}} >0 $ such that
the gradient of the function $\Lagr_{c}(\cdot, \bs{y}) : [0, 2\pi]^{r} \rightarrow \mathbb{R}$ is $L_{c, \bs{y}}$-Lipschitz continuous.
\end{lemma}

\begin{proof}
The proof is given in Appendix~\ref{app:smoothnessL}.
\end{proof}

\begin{lemma}[Lipschitz continuity of $f(\bst)$ and $\bs{A}(\bst)$] \label{lemma:lipconCPhi} The function $f : [0, 2\pi]^{r} \rightarrow \mathbb{R}$, where $f(\bst) = \langle I \otimes C \rangle_{\bst}$, and the linear map $\bs{A}(\bst) : [0, 2\pi]^{r} \rightarrow \mathbb{R}^{M}$ are $L_{f}$-Lipschitz and $L_{\bs{A}}$-Lipschitz continuous, respectively, for some $L_{f},L_{\bs{A}} > 0$.
\end{lemma}
\begin{proof}
The proof directly follows from the Lipschitz continuity of $h(\bst)$ (Lemma~\ref{lemma:lipconexp}).
\end{proof}


\subsection{Inequality Constrained Standard Form (ICSF) of SDPs}

\label{sec:ICSF-reform}

In this section, we consider the following inequality constrained primal form of SDPs:
\begin{equation}
    \label{eq:sdpinequalityform}
\begin{aligned}
p^* = \sup_{X\succcurlyeq 0}&\ \ \operatorname{Tr}[CX] \\
\text{subject to}&\ \  \bs{\Phi}(X) \leq \bs{b}.
\end{aligned}
\end{equation}
Here we take vector forms of the Hermiticity-preserving linear map $\Phi$ and the Hermitian operator $B$ because both have a diagonal form. Furthermore, we assume that the last constraint is the trace constraint on the primal PSD variable $X$. As such, we set $A_{M} = I$ and $\operatorname{Tr}[X] \leq b_{M}$. Additionally, for this case also, we make an assumption that the SDPs are weakly constrained, i.e., $N \gg M$. The corresponding dual form is given as follows:
\begin{equation}
    \label{eq:sdpinequalitydualform}
\begin{aligned}
d^* = \inf_{\bs{y}\geq 0}&\ \ \bs{b}^{\top}\bs{y} \\
\text{subject to}&\ \  \Phi^{\dagger}(\bs{y}) \succcurlyeq C,
\end{aligned}
\end{equation}
where $\bs{y} = (y_{1},\ldots, y_{M})^{\top}$ is a dual variable. We can determine the dual map $\Phi^{\dagger}$ as follows, using the definition of the adjoint of a linear map (see Definition~\eqref{eq:defmap2}):
\begin{align}
\bs{y}^{\top}\bs{\Phi}(X)  &  =\sum_{i=1}^{M-1}y_{i}\operatorname{Tr}
[A_{i}X]+y_{M}\operatorname{Tr}[X]\\
&  =\operatorname{Tr}\!\left[  \left(  \sum_{i=1}^{M-1}y_{i}A_{i}+y_{M}I\right)
X\right]  ,
\end{align}
which implies that
\begin{equation}
\Phi^{\dag}(\bs{y})=\sum_{i=1}^{M-1}y_{i}A_{i}+y_{M}I.
\end{equation}

We derive unconstrained formulations for the primal and dual forms of inequality constrained SDPs. First, by taking the primal form of SDPs into account, we introduce slack variables and convert the inequality constrained problem~\eqref{eq:sdpinequalityform} to the following equality constrained problem:
\begin{equation}
    \label{eq:sdpinequalityslackform}
\begin{aligned}
p^* = \sup_{X\succcurlyeq 0, \bs{z} \geq 0}&\ \ \operatorname{Tr}[CX] \\
\text{subject to}&\ \  \bs{b} - \bs{\Phi}(X) = \bs{z}, 
\end{aligned}
\end{equation}
where $\bs{z} = (z_{1},\ldots, z_{M})^{\top}$ is a vector of slack variables and $z_{1}, \ldots, z_{M} \geq 0$. 
Now, with the problem~\eqref{eq:sdpinequalityslackform} being an equality constrained problem, we can use iVQAEC (see Algorithm~\ref{algo:iVQAEC}) to solve it.

Next, we focus on the dual form of inequality constrained SDPs as follows:
\begin{align}
d^* & = \inf_{\bs{y} \geq 0} \left \{ \bs{b}^{\top}\bs{y}: \Phi^{\dagger}(\bs{y}) \succcurlyeq C \right\}\\
& =   \inf_{y_{1},\ldots,y_{M}\geq 0}\left\{  \sum_{i=1}^{M-1}b_{i}
y_{i}+y_{M}b_{M}:\sum_{i=1}^{M-1}y_{i}A_{i}+y_{M}I\succcurlyeq C\right\} \\
&  =\inf_{y_{1},\ldots,y_{M}\geq 0}\left\{  \sum_{i=1}^{M-1}b_{i}
y_{i}+y_{M}b_{M}:y_{M}I\succcurlyeq C-\sum_{i=1}^{M-1}y_{i}A_{i}\right\}\label{eq:incdualexp}.
\end{align}
Now, consider that
\begin{align}
\inf_{t\geq0}\left\{  t:t I\succcurlyeq H\right\}   =\max
\{0, \lambda_{\max}(H)\} 
 =\max\left\{  0,\sup_{\rho \in \mathcal{D}^{N}}\operatorname{Tr}[H\rho]\right\}  ,
\end{align}
where $H$ is a Hermitian operator, $\lambda_{\max}(H)$ is the maximum eigenvalue of $H$, and the supremum is over the set $\mathcal{D}^{N}$ of density operators. Using this fact, we convert the constrained problem~\eqref{eq:incdualexp} into the following unconstrained problem:
\begin{align}
d^* & = \inf_{y_{1},\ldots,y_{M-1}\geq 0}  \left\{\sum_{i=1}^{M-1}b_{i}y_{i}+b_{M}\cdot
\max\left\{  0,\sup_{\rho \in \mathcal{D}^{N}}\operatorname{Tr}
\left[ \left ( C-\sum_{i=1}^{M-1}y_{i}A_{i} \right) \rho\right]  \right\}\right\}\\
 & =\inf_{\bar{\bs{y}}\geq 0}  \left\{\sum_{i=1}^{M-1}b_{i}y_{i}+b_{M}\cdot
\max\left\{  0,\sup_{\rho \in \mathcal{D}^{N}}\operatorname{Tr}
\left[ H(\bar{\bs{y}})  \rho\right]  \right\}\right\},
\end{align}
where we set
\begin{align}
\bar{\bs{y}} & \coloneqq (y_{1},\ldots,y_{M-1})^\top,
\\ 
H(\bar{\bs{y}}) & \coloneqq  C-\sum_{i=1}^{M-1}y_{i}A_{i}.    
\end{align}
Furthermore, as the first term is independent of $\rho$, we have that
\begin{align}
d^* = \inf_{\bar{\bs{y}} \geq 0} \sup_{\rho \in \mathcal{D}^{N}}\  \left\{\sum_{i=1}^{M}b_{i}y_{i}+b_{M}\cdot
\max\left\{ 0, \operatorname{Tr}
\left[ H(\bar{\bs{y}})  \rho\right]  \right\} \right\}.
\end{align}
We can use Sion's minimax theorem~\cite{Sion1958} and interchange the supremum and infimum because the set $\mathcal{D}^{N}$ of density operators is compact and convex, the objective function is convex with respect to $\bar{\bs{y}}$ for all $\rho \in \mathcal{D}^{N}$, and it is quasi-concave with respect to $\rho \in \mathcal{D}^{N}$ for all $ \bar{\bs{y}} \geq 0$. Hence,
\begin{align}\label{eq:inequnp}
 d^* =  \sup_{\rho \in \mathcal{D}^{N}} \inf_{\bar{\bs{y}} \geq 0}\  \left\{\sum_{i=1}^{M-1}b_{i}y_{i}+b_{M}\cdot
\max\left\{ 0, \operatorname{Tr}
\left[ H(\bar{\bs{y}})  \rho\right]  \right\}\right\}.
\end{align}

The above problem is not differentiable at some parameter values due to the $\max$ operation in the objective function. In order to smoothen the sharp corners that result from this operation, we modify the above problem in the following way by introducing~$\gamma > 0 $:
\begin{align}\label{eq:form2soft}
 d' \coloneq \sup_{\rho \in \mathcal{D}^{N}} \inf_{\bar{\bs{y}} \geq 0}\  \left\{ \sum_{i=1}^{M-1}b_{i}y_{i}+
\frac{b_{M}}{\gamma}\ln \!\left ( e^{\gamma \operatorname{Tr}
\left[ H(\bar{\bs{y}})  \rho\right]} + 1\right )  \right\}.
\end{align}
In the above, we used the scaled version of the log-sum-exp function~\cite[Section~3.1.5]{BV04} as an approximation to the original $\max$ function in \eqref{eq:inequnp}. This approximation can be controlled by varying the parameter $\gamma$ because the following holds for all $x, y \in \mathbb{R}$:
\begin{equation}
\max\{x, y\} \leq \frac{\ln(e^{\gamma x} + e^{\gamma y})}{\gamma} \leq \max\{x, y\} + \frac{\ln(2)}{\gamma}.
\end{equation}
By using this approximation, the objective function in~\eqref{eq:form2soft} does not have any sudden corners where the partial derivatives can change drastically. In fact, the objective function of~\eqref{eq:form2soft} is infinitely differentiable. It is also equivalent to the previous objective function in~\eqref{eq:inequnp} in the limit $\gamma \rightarrow \infty$. For all finite values of $\gamma$, the value $d^*$ is bounded from above by $d'$. 

\subsubsection{Variational Quantum Algorithm for SDPs in ICSF}

\label{sec:vqa-icsf}

Solving the unconstrained optimization problem~\eqref{eq:form2soft} using a gradient based classical technique is computationally expensive if the dimension $N$ of the operators is large. Hence, in this section, we propose a variational quantum algorithm to solve this problem.

First, we introduce a parameterization of density operators using parameterized quantum circuits, in order to efficiently estimate the expectation values of the Hermitian operators $C, A_{1}, \ldots, A_{M-1}$. Second, we optimize the resulting objective function using our variational quantum algorithm.\\

\noindent
\textbf{Parameterization:} 
Let $U_{RS}(\bst)$ be a parameterized quantum circuit with a parameter $\bst \in [0, 2\pi]^{r}$, and let it act on the all-zeros state of the quantum system $RS$ and prepare a purification of the density operator $\rho_{S}(\bst)$. The structure of $U_{RS}(\bst)$ is given by~\eqref{eq:vqauni}. Furthermore, the following equality holds because $U_{RS}(\bst)$ generates a purification of $\rho_{S}(\bst)$:
\begin{equation}\label{eq:eqexp-2}
    \operatorname{Tr}[H \rho_{S}(\bst)] = \langle \bs{0} |_{RS} U_{RS}^{\dagger}(\bst) \left (I_R \otimes H_S \right)  U_{RS}(\bst) | \bs{0}\rangle_{RS} = \langle I\otimes H \rangle_{\bs{\theta}},
\end{equation}
where $H \in \{C, A_{1},\ldots,A_{M-1}\}$. 

Due to the parameterization of the density operator $\rho_{S}(\bst)$, the problem~\eqref{eq:form2soft} transforms into the following optimization problem:
\begin{align}
\label{eq:form2paramsoft-later}
d' \coloneqq  \sup_{\bs{\theta} \in [0, 2\pi]^{r}} \inf_{\bar{\bs{y}} \geq 0}\   \mathcal{F}_{\gamma}(\bst, \bar{\bs{y}})  ,
\end{align}
where
\begin{equation}
    \mathcal{F}_{\gamma}(\bst, \bar{\bs{y}})\coloneqq  \sum_{i=1}^{M-1}b_{i}y_{i}+
\frac{b_{M}}{\gamma}\ln \!\left ( e^{\gamma \langle I \otimes H(\bar{\bs{y}})\rangle_{\bst}} + 1\right )
\end{equation}
and we optimize over the space of quantum circuit parameters $\bst \in [0, 2\pi]^{r}$.
As mentioned before, we assume that the objective function $ \mathcal{F}_{\gamma}(\bst, \bar{\bs{y}})$ is faithful, which means that the global optimal value of the optimization problem~\eqref{eq:form2paramsoft-later} is equal to $d'$.

\begin{algorithm}[t]
\caption{\texttt{iVQAIC}($C, \{A_{i}\}_{i=1}^{M}, \{b_{i}\}_{i=1}^{M}, \eta, \epsilon, \gamma$)}\label{algo:iVQAIC}
\begin{algorithmic}[1]\onehalfspacing
\STATE \textbf{Input:} Hermitian operators $(C, \{A_{i}\}_{i=1}^{M})$, scalars $\{b_{i}\}_{i=1}^{M}$, learning rate $\eta > 0$, precision $\epsilon > 0$, constant $\gamma > 0$. 
\vspace{2pt}
\STATE \textbf{Initialization:} $\bar{\bs{y}}^{1} = \bs{0}$.\\
\vspace{0.7em}
\textcolor{gray}{\textit{\# For any step, expectation values of observables and their gradients are evaluated using a parameterized quantum circuit.}}\\
\vspace{0.7em}
\FOR{$k = 1, 2, \ldots,$}
\vspace{0.7em}
\STATE Maximize $\mathcal{F}_{\gamma}(\bst, \bar{\bsy}^{k})$, where $\bar{\bsy}^{k}$ is fixed, to obtain $\bst^{k+1}$ such that the following holds:
\begin{equation}
     \left \Vert \nabla\mathcal{F}_{\gamma}(\bst^{k+1}, \bar{\bs{y}}) \right \Vert \leq \epsilon.\nonumber
\end{equation}
\STATE $\bar{\bs{y}}^{k+1} = \bar{\bs{y}}^{k}  - \eta \nabla_{\bar{\bsy}} \mathcal{F}_{\gamma}(\bst^{k+1}, \bar{\bsy}^{k})$
\vspace{0.7em}
\IF{$ \left \Vert \nabla\mathcal{F}_{\gamma}(\bst^{k+1}, \bar{\bs{y}}^{k+1}) \right \Vert \leq \epsilon$}
\vspace{0.5em}
\STATE STOP and return $\mathcal{F}_{\gamma}(\bst^{k+1}, \bar{\bs{y}}^{k+1})$
\ENDIF
\ENDFOR
\end{algorithmic}
\end{algorithm}

The objective function of~\eqref{eq:form2paramsoft-later} is in general non-convex with respect to the quantum circuit parameter $\bst$. On the contrary, it is concave in $\bar{\bs{y}}$. Hence, the optimization problem~\eqref{eq:form2paramsoft-later} is a non-convex--concave optimization problem. For such a setting, finding a globally optimal solution is generally NP-hard~\cite[Section 2.1]{Danilova2022}. Therefore, we focus on obtaining approximate stationary points of~\eqref{eq:form2paramsoft-later}. As a global optimal point is also a stationary point, if we use techniques to initialize quantum circuit parameters according to the problem at hand such that these parameters lie in the vicinity of a global optimal point, then our algorithm will converge to that point.

\begin{figure}
\centering
\includegraphics[width=\textwidth]{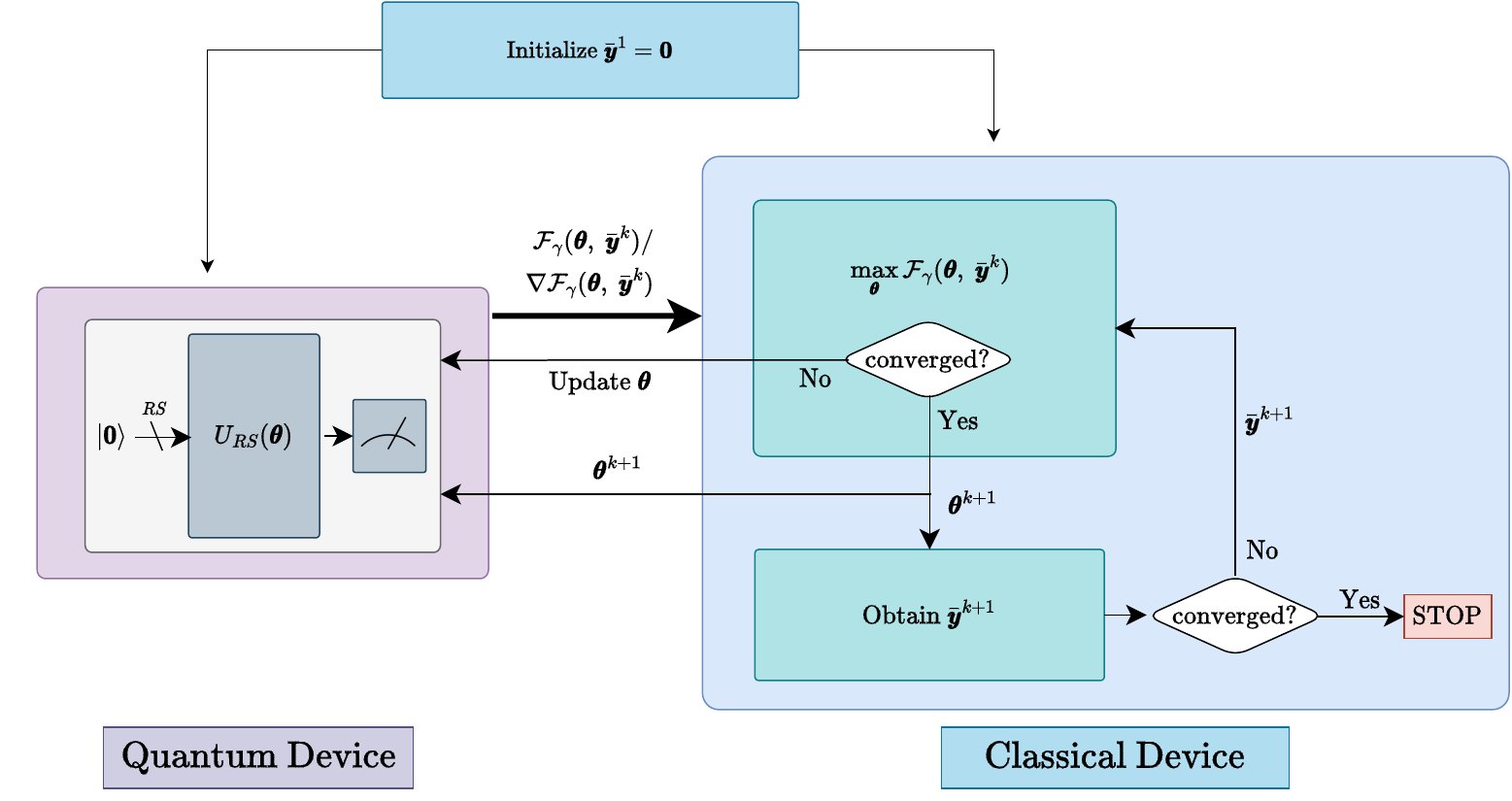}
\caption{This figure depicts the iVQAIC algorithm in which we utilize one parameterized quantum circuit, i.e., $U_{RS}(\bst)$.}
\label{fig:iVQAIC}
\end{figure}

\begin{definition}[First order $\epsilon$-stationary points of~\eqref{eq:form2paramsoft-later}]
For $\epsilon>0$,
a point $(\bs{\theta}, \bar{\bs{y}}) \in [0, 2\pi]^{r} \times \mbbR^{M-1}$ is a first order $\epsilon$-stationary point of~\eqref{eq:form2paramsoft-later} if the following holds:
\begin{equation}
    \left \Vert \nabla\mathcal{F}_{\gamma}(\bst, \bar{\bs{y}}) \right \Vert \leq \epsilon.
\end{equation}
\end{definition}

We use the above condition as a stopping criterion for our variational quantum algorithm.

In order to obtain a first order $\epsilon$-stationary point of~\eqref{eq:form2paramsoft-later}, we propose a variational quantum algorithm (see Algorithm~\ref{algo:iVQAIC}).  We call this algorithm \textit{inexact Variational Quantum Algorithm for Inequality Constrained standard form} (iVQAIC). It is an inexact version because we solve the subproblem involving the maximization to an approximate stationary point instead of solving it until global optimality is reached, due to the non-convex nature of the objective function in terms of quantum circuit parameters. Its pseudocode is provided in Algorithm~\ref{algo:iVQAIC}, and it is depicted in Figure~\ref{fig:iVQAIC}.

We run iVQAIC on a classical computer and utilize a parameterized quantum circuit $U_{RS}(\bst)$ with a parameter $\bst \in [0, 2\pi]^r $ for estimating the expectation values of the Hermitian operators. Moreover, the partial derivative of $\mathcal{F}_{\gamma}(\bst, \bar{\bs{y}})$ with respect to the quantum circuit parameters $\bst$ depends on the partial derivatives of the expectation values of the Hermitian operators $I\otimes C$, $I\otimes A_{1}$, \ldots, $I \otimes A_{M-1}$ with respect to $\bst$. In order to compute the partial derivatives of the latter with respect to $\bst$, we use the parameter-shift rule (see~\eqref{eq:parashift}). Note that we do not explicitly mention these quantum-circuit calls in the pseudocode of the algorithm.

For this case, we do not state a theorem that indicates the convergence rate for finding approximate stationary points of~\eqref{eq:form2paramsoft-later} in terms of the number of iterations in the for-loop of the algorithm. Rather, we prove a property of the objective function that is necessary for providing such a convergence analysis, i.e., smoothness of $\mathcal{F}(\cdot, \bar{\bs{y}})$ for a fixed $\bar{\bs{y}} \geq 0$. We state it formally in the following lemma:

\begin{lemma}[Smoothness of $\mathcal{F}_{\gamma}(\cdot, \bar{\bs{y}})$] \label{lemma:smoothnessF}
For all $\bar{\bs{y}} \geq 0$,
there exists $L_{\gamma, \bar{\bs{y}}}>0$ such that the gradient of the function $\mathcal{F}_{\gamma}(\cdot, \bar{\bs{y}}) : [0, 2\pi]^{r} \rightarrow \mathbb{R}$ is $L_{\gamma, \bar{\bs{y}}}$-Lipschitz continuous.
\end{lemma}

\begin{proof}
The proof is given in Appendix~\ref{app:smoothnessF}.
\end{proof}

\section{Numerical Simulations}
\label{sec:simulations}

For validating our reformulations of SDPs and verifying the convergence of their respective algorithms to approximate stationary points, we randomly generated SDPs such that they contain a valid feasible region. To verify the convergence of the algorithms proposed in this paper, we focus on three cases based on the number $M$ of constraints and the dimension~$N$ of the input Hermitian operators, and whether the problem is an equality or inequality constrained problem:

\begin{enumerate}
    \item $N \approx M$: Here, we consider a well-known and extensively studied MaxCut problem. First, we briefly recall what the problem is and its SDP relaxation in the next subsection. This problem is taken into account because the number of constraints and the dimension of the matrices are equal ($N=M$). For this case, we evaluated the performance of iVQAGF (see Algorithm~\ref{algo:iVQAGF}), as it does not make the  weakly-constrained assumption on SDPs. Figure~\ref{fig:plot-1} shows the convergence of iVQAGF for solving randomly generated SDP instances of the MaxCut problem. Furthermore, we performed these simulations for different dimensions of the input matrices. Specifically, we considered $ N \in \{8, 16, 32\}$.

    \item $N \gg M$: We divide this case further according to the type of constraints: (2a) equality constraints and (2b)  inequality constraints. For randomly generated equality-constrained problems, we report the performance of iVQAEC (see Algorithm~\ref{algo:iVQAEC}). Similarly, we analyse the performace of iVQAIC  (see Algorithm~\ref{algo:iVQAIC}) for randomly generated instances of an inequality constrained problem. Here also we consider $ N \in \{8, 16, 32\}$. Figure~\ref{fig:plot-2} and Figure~\ref{fig:plot-3} showcase the convergence of iVQAEC and iVQAIC, respectively, for randomly generated instances of their respective problems. 
\end{enumerate}

We take into account Assumption~\ref{as:hermiop} while creating the input Hermitian matrices. We assume that the Pauli string decomposition of these input matrices are provided beforehand. Additionally, for the analysis of iVQAGF for solving MaxCut, we assume that we are provided with the Pauli string decomposition of $C^{\top}$ and the Choi operator of the linear map $\Phi$, i.e., $\Gamma^{\Phi}$.

We executed our algorithms using Pennylane's Python libraries where we set QASM simulator of the Qiskit Python package as a backend.\footnote{The source code is available \href{https://github.com/Dhrumil2910/Variational-Quantum-Algorithms-for-Semidefinite-Programming}{here.}} Pennylane is an open-source Python library developed by Xanadu for differential programming of quantum computers~\cite{Pennylane}. Similarly, Qiskit is an open-source package/interface developed by IBM to interact with the underlying quantum computer~\cite{Qiskit}. Additionally, the QASM simulator simulates a real IBM Quantum Backend, which is actually noisy in nature due to gate errors and decoherence. We integrated Pennylane and Qiskit's QASM simulator using the Pennylane-Qiskit plugin.

We use the Strong Entangling Layers template of Pennylane as our variational ansatz, where each layer consists of $O(\text{poly}(n))$  single-qubit rotations and  entangling gates. We then repeat this layer $O(\text{poly}(n))$ number of times, where $n$ is the number of qubits. Therefore, the overall gate complexity is $O(\text{poly}(n))$.

In order to assess the convergence of our methods to a globally optimal point, we initialize the quantum circuit parameter such that it lies in the convex region of that globally optimal point.
We then report the results and their associated analyses on how well our algorithms perform on the QASM noisy simulator and compare the results with a noiseless simulator. We also report the time complexity (number of iterations of the for-loops) of our algorithms, i.e., how fast our algorithms converge to an actual globally optimal point. 

First, for the sake of completeness, we recall the definition of a cut and a MaxCut of a given graph.
\begin{figure}
    \centering
    \subfloat{\includegraphics[width=6.5cm]{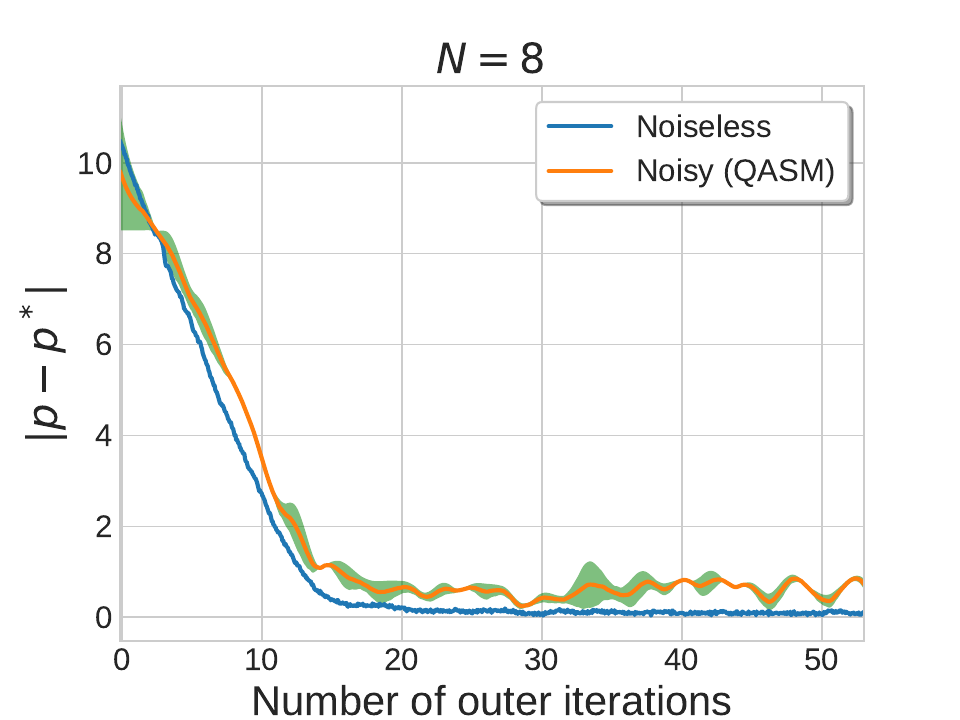} }
    \qquad
    \subfloat{\includegraphics[width=6.5cm]{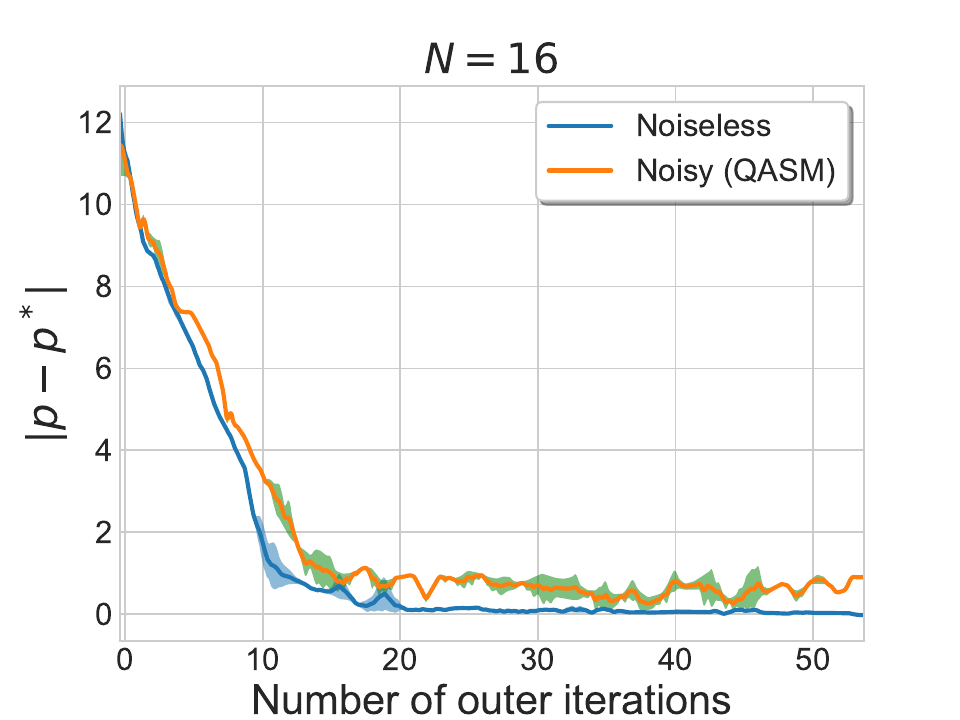} }
    \newline
    \subfloat{\includegraphics[width=6.5cm]{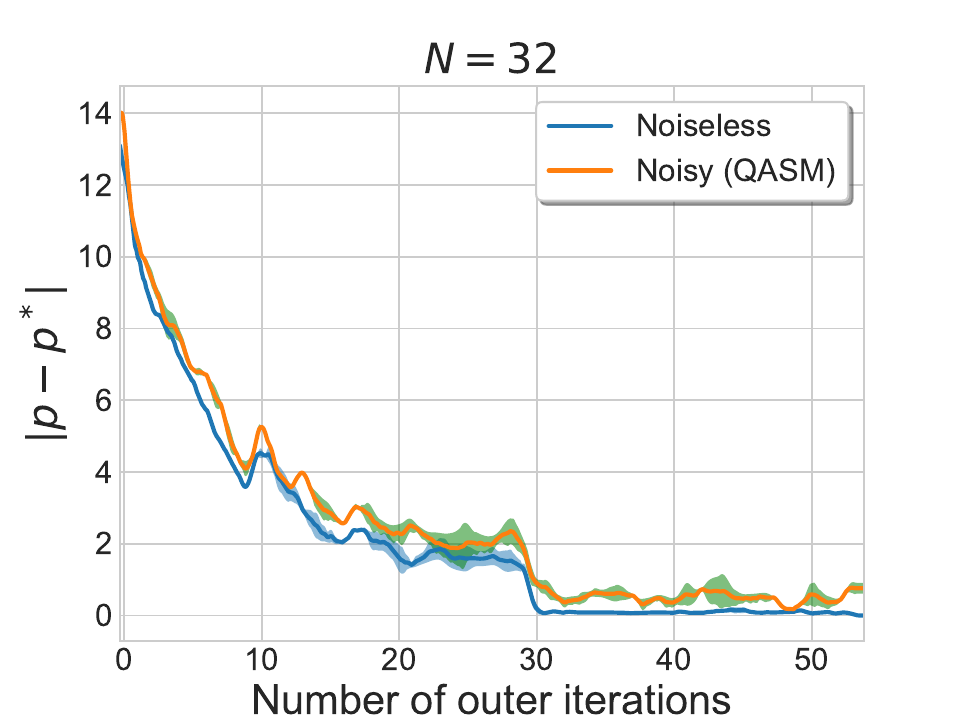} }
    \caption{Convergence of iVQAGF for three randomly generated MaxCut-SDP instances with different numbers of vertices in the graph: $ N \in \{8, 16, 32\}$.}
    \label{fig:plot-1}
\end{figure}

\begin{figure}
    \centering
    \subfloat{{\includegraphics[width=6.5cm]{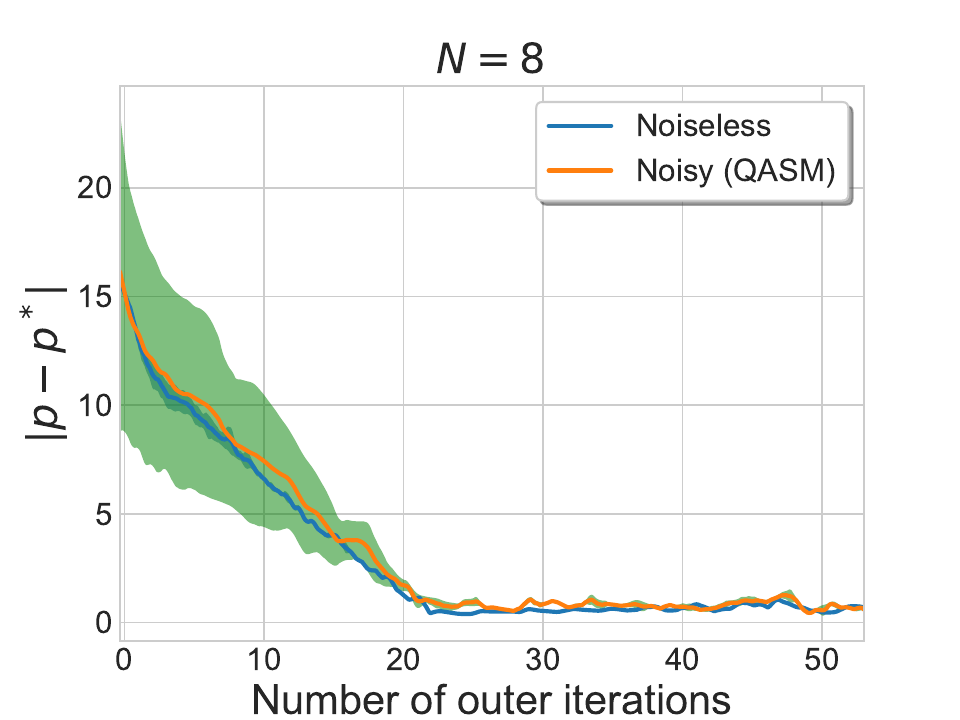} }}
    \qquad
    \subfloat{{\includegraphics[width=6.5cm]{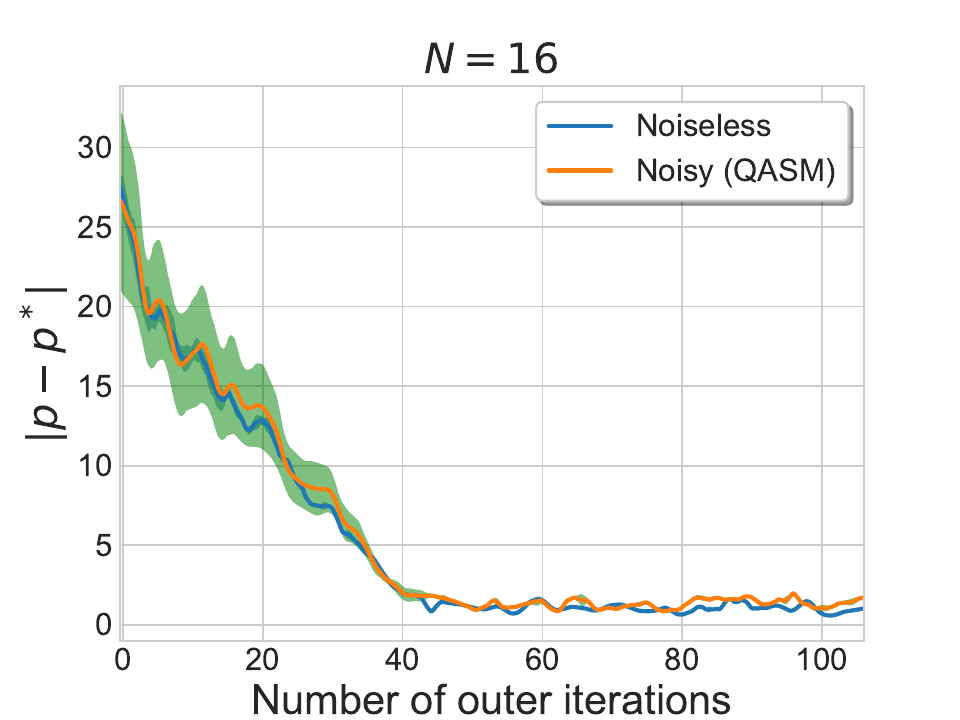} }}
    \newline
    \subfloat{{\includegraphics[width=6.5cm]{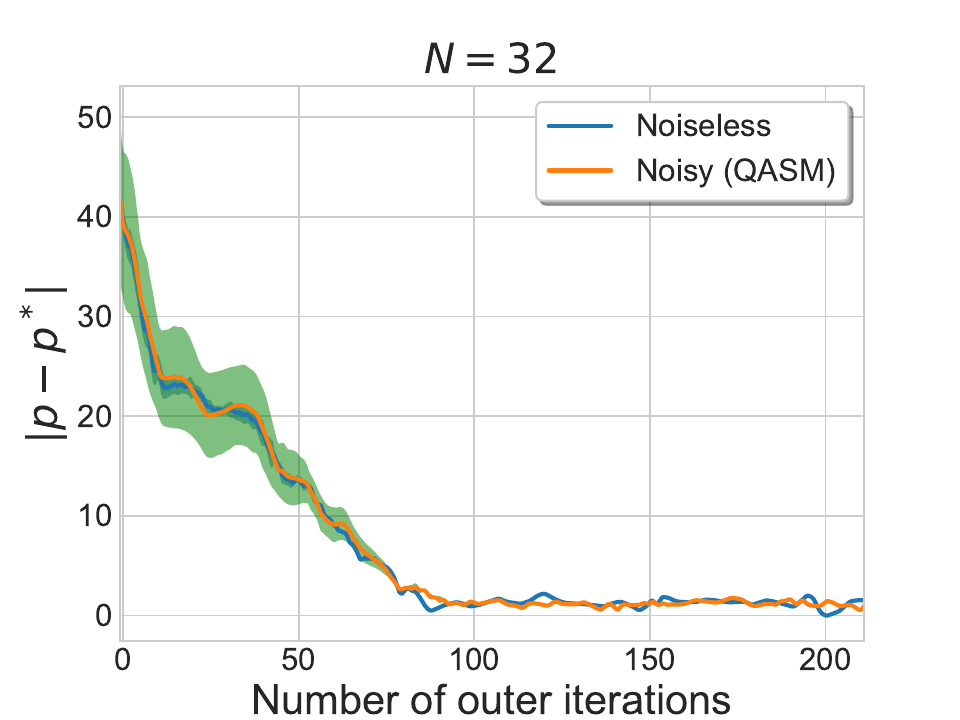} }}
    \caption{Convergence of iVQAEC for three separate cases of randomly generated equality constrained semidefinite programs and nonempty feasbile regions: $ N \in \{8, 16, 32\}$.}
    \label{fig:plot-2}
\end{figure}

\begin{figure}
    \centering
    \subfloat{{\includegraphics[width=6.5cm]{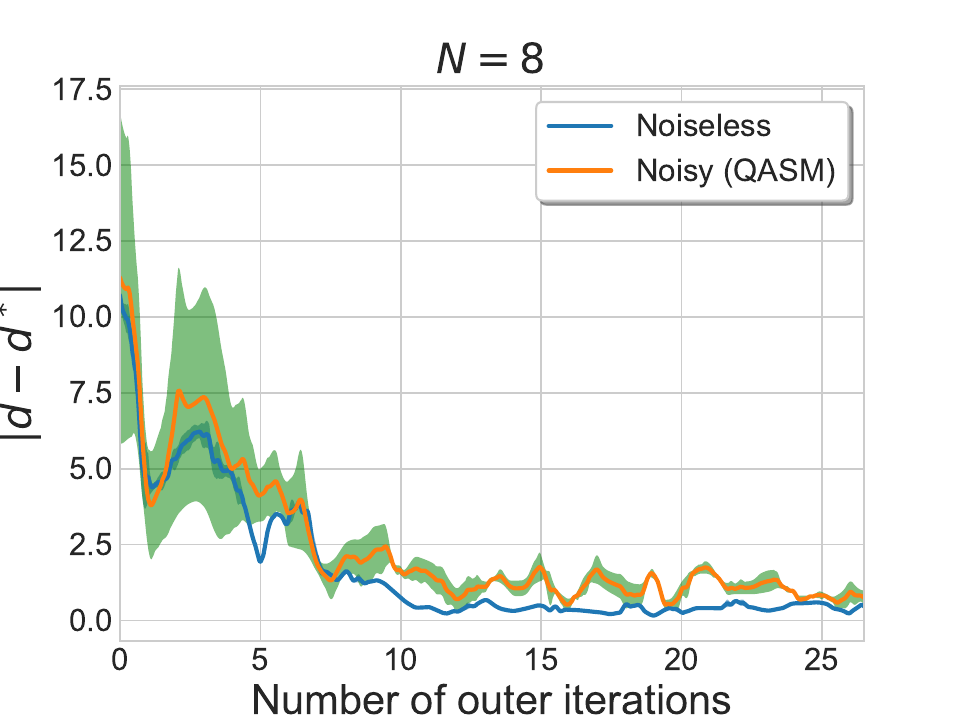} }}
    \qquad
    \subfloat{{\includegraphics[width=6.5cm]{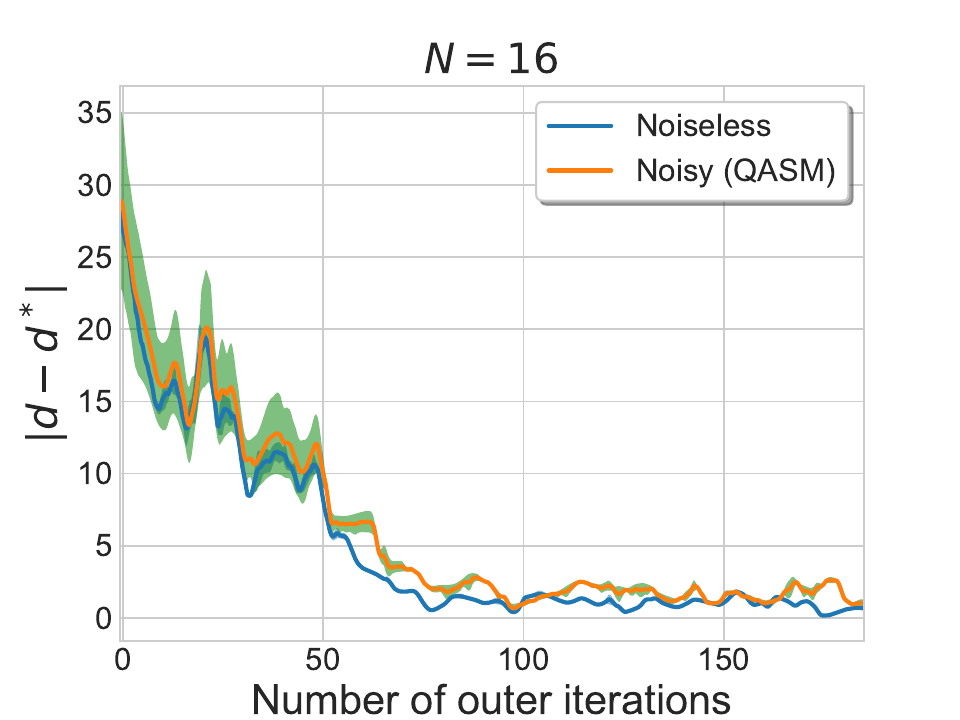} }}
    \newline
    \subfloat{{\includegraphics[width=6.5cm]{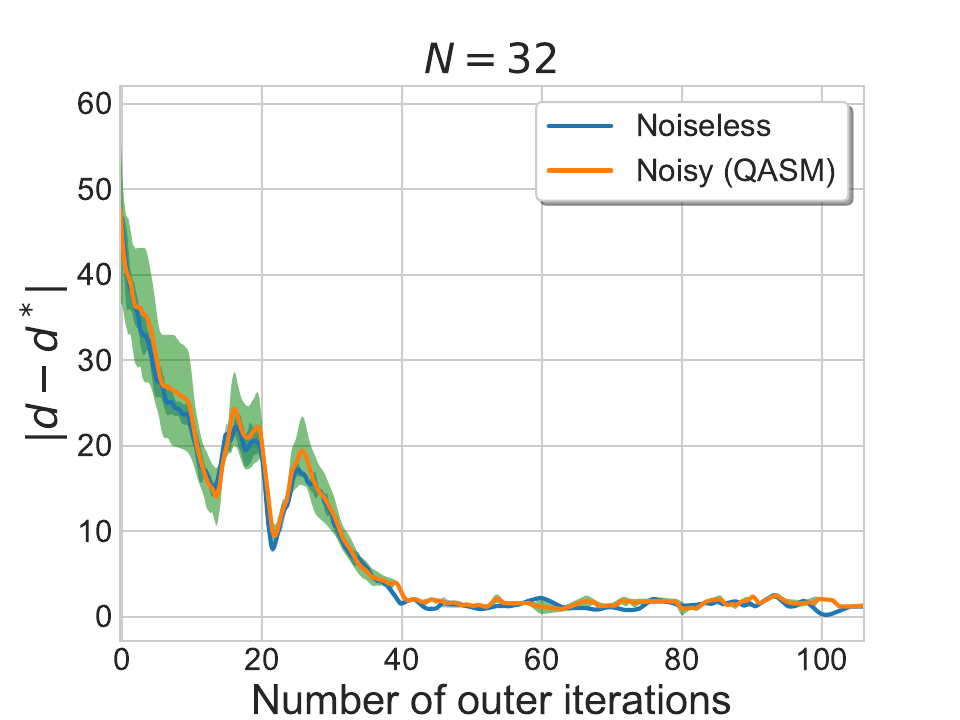} }}
    \caption{Convergence of iVQAIC   for three separate cases of randomly generated inequality constrained semidefinite programs with nonempty feasbile regions: $ N \in \{8, 16, 32\}$.}
    \label{fig:plot-3}
\end{figure}
\begin{definition}[Cut and MaxCut]
A cut is a bi-partition $W$ of the vertex set $V$ of a graph $G = (V, E)$, where $|V| = N$. An edge $(i, j)\in E$ is part of the cut set if its vertices $i$ and $j$ lie in separate partitions. A MaxCut is the largest cut possible of a graph G.
\end{definition}

The problem of finding a MaxCut of a graph can be formulated as a Quadratic Integer Program (QIP)~\cite{MP90}:
\begin{equation}
    \label{eq:maxcutqip}
\begin{aligned}
 [\text{QIP}]:\quad & \sup\ \ \sum_{(i, j) \in E} \frac{1}{4}(x_{i} - x_{j})^2 \\
 &\text{subject to}\ \  x_{i} \in \{-1, 1\};\ \forall i\in V.
\end{aligned}
\end{equation}
Here, the optimal value of the above optimization is the optimal cut size (i.e., the number of the edges in the MaxCut). This quadratic integer program is in general computationally intractable~\cite{PY91}. However, there exist LP and SDP relaxations for the above program~\cite{GW95}. 
\begin{equation}
    \label{eq:maxcutlp}
\begin{aligned}
 [\text{LP}]:\quad & \sup  \sum_{(i, j) \in E} \frac{1}{4}\Vert \bs{v}_{i} - \bs{v}_{j}\Vert^2  \\
 &\text{subject to}\ \  \Vert \bs{v}_{i} \Vert^{2} = 1, \bs{v}_i \in \mathbb{R}^{N};\ \forall i\in V.
\end{aligned}
\end{equation}
The above mentioned LP formulation is not good enough as it is a 1/2-approximation to the original problem. It is well known that this can be extended to an SDP formulation in which an algorithm proposed in \cite{GW95} gives a 0.879-approximation of the original QIP in~\eqref{eq:maxcutqip}:
\begin{equation}
    \label{eq:maxcutsdp}
\begin{aligned}
[\text{SDP}]: \quad & \sup \sum_{(i, j) \in E} \frac{1}{4}(X_{ii} - 2 X_{ij} + X_{jj})  \\
 &\text{subject to}\ \  X \succcurlyeq 0, \ X_{ii} = 1, \  \forall i\in V.
\end{aligned}
\end{equation}

For our case, we numerically simulate iVQAGF for solving the aforementioned SDP relaxation of a MaxCut problem. Recasting the constraints of the above MaxCut SDP as constraints of the general form of SDP, we can write them as $\Phi(X) = B$, where
\begin{equation}
    B = I, \qquad \Phi(X) = \text{diag}\left( \operatorname{Tr}[A_{1}X],\ldots, \operatorname{Tr}[A_{M}X] \right),\\
\end{equation}
and $A_{i} = |i\rangle \langle i |$. Therefore, the Choi operator for this linear map is given as
\begin{equation}
    \Gamma^{\Phi} =  \sum_{i,j=0}^{N-1}|i\rangle \langle j|\otimes\Phi(|i\rangle \langle
j|)
 = \sum_{i = 0}^{N-1}|i\rangle \langle i|\otimes |i\rangle \langle
i|.
\end{equation}
Additionally, the objective function can be written as $\operatorname{Tr}[CX]$, where $C = L/4$ and $L$ is the Laplacian matrix for the graph $G$.

The Pauli decomposition of the above Choi operator consists of $N = 2^n$ Pauli strings. This gives the impression that we need to compute $2^n$ expectation values in order to compute the expectation value of $\Gamma^{\Phi}$. However, all these Pauli strings can be constructed from just $I$ and $\sigma_{z}$ operators, because it is clearly a diagonal matrix. This implies that all the Pauli strings are mutually commuting, and their expectation values can be estimated simultaneously~\cite{crawford2021efficient}. Therefore, we just need to compute a single expectation value to evaluate the expectation value of $\Gamma^{\Phi}$.

The shaded regions in Figure~\ref{fig:plot-1}, \ref{fig:plot-2}, and \ref{fig:plot-3} signify the variance in the values for different runs (specifically 20) of the algorithm. The $y$-axis measures the difference between the cost function value at each iteration and the actual optimal value evaluated with the CVXPY package~\cite{cvxpy}. The $x$-axis shows the number of outer iterations of algorithms, i.e., the time needed to converge to the actual optimal value.

The numerical simulations demonstrate that all three algorithms indeed converge to their respective optimal values approximately. This numerical evidence suggests that all three proposed algorithms work well in practice. In other words, convergence to the optimal parameters in the presence of noise showcases noise resilience of our variational quantum algorithms.

\section{Conclusion}

In this paper, we proposed variational quantum algorithms for solving semidefinite programs. We considered three constrained formulations of SDPs, which were first converted to unconstrained forms by employing a series of reductions. When the dimension~$N$ of the input Hermitian operators of SDPs is large and for these unconstrained forms, the computation of the objective function's gradient is difficult when using known classical techniques. To address this problem, we utilized parameterized quantum circuits to estimate these gradients. We  also established the convergence rate and total iteration complexity of one of our proposed VQAs. Finally, we numerically simulated our variational quantum algorithms for different instances of SDPs, and the results of these simulations provide evidence that convergence still occurs in noisy settings.

The estimation of the gradients using parameterized quantum circuits is stochastic in nature. In this paper, we assumed that we have unbiased estimators of these gradients, and the variance of these estimators is also small. Therefore, it remains open to study the effect of the variance of these estimators on the convergence rate of our algorithms. 

\begin{acknowledgments}
DP and MMW acknowledge support from the National Science
Foundation under Grant No.~1907615. PJC acknowledges initial support from the Los Alamos National Laboratory (LANL) ASC Beyond Moore's Law project, and later support from the U.S. Department of Energy (DOE), Office of Science, Office of Advanced Scientific Computing Research, under the Accelerated Research in Quantum Computing (ARQC) program. DP acknowledges Prof.~Rahul Shah's helpful suggestions and support from the National Science Foundation  under Grant No.~2137057.
\end{acknowledgments}

\section*{Author Contributions}

The following describes the different contributions of all authors of this work, using roles defined by the CRediT
(Contributor Roles Taxonomy) project \cite{NISO}:

\medskip 

\noindent \textbf{DP}: Conceptualization, Methodology, Software, Validation, Formal Analysis, Investigation, Data Curation, Writing - Original Draft, Writing - Review \& Editing, Visualization.

\medskip 
\noindent \textbf{PC}: Writing - Review \& Editing.

\medskip 
\noindent \textbf{MMW}: Conceptualization, Methodology,  Formal Analysis,  Writing - Review \& Editing, Supervision, Project administration, Funding acquisition.

\bibliographystyle{quantum}
\bibliography{Ref}

\appendix

\section{Appendix}

\subsection{Proof of Lemma~\ref{lemma:lipmul}}

\label{app:lipmul}

The proof is rather straightforward. First, let us consider a function with two parameters $x$ and $y$, so that $f : \mathbb{R}^{2} \rightarrow \mathbb{R}$. Suppose that the function $(\cdot) \to f(\cdot, y)$ is Lipschitz continuous with Lipschitz constant $L_{X}$, for all $y\in \mathbb{R}$, and suppose that the function $(\cdot) \to f(x, \cdot)$ is Lipschitz continuous with Lipschitz constant $L_{Y}$, for all $x\in \mathbb{R}$. Therefore, the following holds according to the definition of Lipschitz continuity (recall Definition~\ref{eq:lipcongen}):
\begin{align}\label{eq:lxly-1}
    \left | f(x, y) - f(x', y)\right | & \leq L_{X} \left |x - x'\right|\quad \forall x,x',y\in \mathbb{R}, \\
    \left | f(x, y) - f(x, y')\right | & \leq L_{Y} \left |y - y'\right|\quad \forall x,y,y'\in \mathbb{R}.
    \label{eq:lxly-2}
\end{align}
Now consider the following:
\begin{align}
    \left | f(x, y) - f(x', y')\right | & = \left | f(x, y) - f(x', y) + f(x', y) - f(x', y')\right | \\ 
    & \overset{\mathrm{(a)}}{\leq} \left | f(x, y) - f(x', y) \right | + \left | f(x', y) - f(x', y')\right | \\ 
    & \overset{\mathrm{(b)}}{\leq} L_{X} \left |x - x'\right| +  L_{Y} \left |y - y'\right| \\ 
    & \leq \max \{L_{X}, L_{Y}\} \left( \left |x - x'\right| + \left |y - y'\right| \right) \\ 
    & \overset{\mathrm{(c)}}{\leq} \sqrt{2} \max \{L_{X}, L_{Y}\} \lV (x, y) - (x', y') \rV,
\end{align}
where the inequality $(a)$ follows from the triangle inequality, inequality $(b)$ follows from~\eqref{eq:lxly-1}--\eqref{eq:lxly-2}, and the last inequality $(c)$ follows from the fact that $\lV \cdot \rV_{1} \leq \sqrt{2} \lV \cdot \rV$ for the two-variable case. Therefore, $L = \sqrt{2} \max \{L_{X}, L_{Y}\}$ is a Lipschitz constant for $f$.

The  proof given above for two variables can be easily extended to a function $f$ with an $n$-variable input, using the fact that $\lV \cdot \rV_{1} \leq \sqrt{n} \lV \cdot \rV$,  where $L = \sqrt{n} \max_{i} \{L_{i}\}_i$ is a Lipschitz constant given in terms of
\begin{equation}
    L_{i} = \sup_{\bs{x}} \left| \frac{\partial f(\bsx)}{\partial x_{i}} \right|.
\end{equation}
Here, $\bs{x} = (x_{1}, \ldots, x_{n})^{\top}$.

\subsection{Proof of Lemma~\ref{lemma:lipvec}}

\label{app:lipvec}

By hypothesis, each component $f_{i}$ of the vector-valued function $f : \mathbb{R}^{n} \rightarrow \mathbb{R}^{m}$ is $L_{i}$-Lipschitz continuous. Hence, the following holds for all $ \bs{x}, \bs{x'} \in \mathbb{R}^{n}$ and $i \in \{1,\ldots,m\}$:
\begin{equation}
    \left | f_{i}(\bs{x}) - f_{i}(\bs{x'})\right | \leq L_{i} \lV \bs{x} - \bs{x'} \rV.
\end{equation}
Now consider the following:
\begin{align}
    \lV f(\bs{x}) - f(\bs{x'})\rV^2 & = \sum_{i=1}^{m}  \left | f_{i}(\bs{x}) - f_{i}(\bs{x'})\right |^{2} \\
    & \leq \sum_{i=1}^{m}  L_{i}^{2} \lV \bs{x} - \bs{x'} \rV^{2} \\ 
    & = \left (\sum_{i=1}^{m}  L_{i}^{2} \right) \lV \bs{x} - \bs{x'} \rV^{2}.
\end{align}
Hence, $L = \sqrt{\sum_{i=1}^{m} L_{i}}$ is a Lipschitz constant for the vector-valued function $f$.

\subsection{Proof of Lemma~\ref{lemma:lipconexp}}

\label{app:lipconexp}

According to Lemma~\ref{lemma:lipmul}, a Lipschitz constant for the multivariate function $h(\bst) = \langle I \otimes O \rangle_{\bst}$ is as follows:
\begin{align}\label{eq:lipexp}
    L_{h} & = \sqrt{r} \max_{i} \left \{\sup_{\bst} \left| \frac{\partial h(\bst)}{\partial \theta_{i}} \right| \right \}_i\\ 
    & \overset{\mathrm{(a)}}{=} \sqrt{r} \max_{i}  \left \{\sup_{\bst} \left| 2 \operatorname{Re} \Bigg[ \langle \bs{0} | U^{\dagger}(\bst) (I \otimes O) \frac{\partial U(\bst)}{\partial \theta_{i}} | \bs{0} \rangle \Bigg] \right |\right \}_i\\ 
    & \overset{\mathrm{(b)}}{\leq} 2 \sqrt{r} \max_{i} \left \{\sup_{\bst} \left|\langle \bs{0} | U^{\dagger}(\bst) (I \otimes O) U_{f}(\bst)e^{-i\theta_{i} H_{i}}(-iH_{i})U_{b}(\bst) | \bs{0} \rangle \right | \right \}_i \\ 
    & \overset{\mathrm{(c)}}{\leq} 2 \sqrt{r} \max_{i} \left \{\sup_{\bst}  \lV U^{\dagger}(\bst) (I \otimes O) U_{f}(\bst)e^{-i\theta_{i} H_{i}}(-iH_{i})U_{b}(\bst) \rV \right \}_i \\ 
    & \overset{\mathrm{(d)}}{\leq} 2 \sqrt{r} \max_{i} \left \{\sup_{\bst} \lV U^{\dagger}(\bst)\rV\lV  O \rV\lV U_{f}(\bst)\rV\lV e^{-i\theta_{i} H_{i}}\rV\lV -iH_{i}\rV\lV U_{b}(\bst)\rV \right \}_i\\ 
     & \overset{\mathrm{(e)}}{\leq}   2 \sqrt{r} \lV O\rV  \max_{i} \left \{ \lV H_{i}\rV \right \}_i.
\end{align}
Here, equality (a) is a consequence of the following fact:
\begin{equation}\label{eq:chainruleexp}
    \frac{\partial h(\bst)}{\partial \theta_{i}} = \langle \bs{0} | U^{\dagger}(\bst) (I \otimes O) \frac{\partial U(\bst)}{\partial \theta_{i}} | \bs{0} \rangle + \langle \bs{0} | \frac{\partial U^{\dagger}(\bst)}{\partial \theta_{i}} (I \otimes O) U(\bst) | \bs{0} \rangle,
\end{equation}
where we used the product rule.
Inequality (b) follows from the fact that $\left|\operatorname{Re}[z]\right| \leq \left| z \right|$ for all  $z \in \mathbb{C}$ . Additionally, we assume $U(\bst) = U_{f}(\bst)e^{-i\theta_{i}H_{i}}U_{b}(\bst)$, so that
\begin{equation}
    \frac{\partial U(\bst)}{\partial \theta_{i}} = U_{f}(\bst)e^{-i\theta_{i} H_{i}}(-iH_{i})U_{b}(\bst).
\end{equation}
Inequality (c) follows from the fact that $\lV A\rV = \sup_{|\psi\rangle, |\phi\rangle } \{|\langle \psi | A | \phi \rangle| :  \lV |\psi\rangle \rV = \lV |\phi\rangle \rV = 1 \}$. Inequality (d) follows from two properties of the spectral norm of a matrix: $\ \Vert AB\Vert \leq \Vert A \Vert \Vert B \Vert$ and $\Vert A\otimes B\Vert \leq \Vert A \Vert \Vert B \Vert$. The inequality $(e)$ follows from the fact that $ \Vert V \Vert = 1$ for every unitary $V$. Similarly, we can evaluate $L_{\nabla h}$ because each component of $\nabla h$, i.e., $\frac{\partial h(\bst)}{\partial \theta_{i}}$ is Lipschitz continuous. This is because according to the parameter-shift rule, we can write $\frac{\partial h(\bst)}{\partial \theta_{i}}$ in terms of the linear combination of $h(\bst + (\pi/4)\hat{\bse}_{i})$ and $h(\bst - (\pi/4)\hat{\bse}_{i})$.

\subsection{Proof of Lemma~\ref{lemma:smoothnessL}}

\label{app:smoothnessL}

The multivariate vector-valued function $\nabla_{\bs{\theta}}\Lagr_{c}(\bs{\theta}, \bs{y})$, where $\Lagr_{c}(\bs{\theta}, \bs{y})$ is defined in \eqref{eq:aug-lagr-def}--\eqref{eq:aug-lagr-def-2}, is $L_{c, \bs{y}}$-Lipschitz continuous for a fixed $\bs{y} \in \mbbR^{M}$ if all its components, i.e.,  $\left\{\frac{\partial \Lagr_{c}(\bs{\theta}, \bs{y})}{\partial \theta_{i}} \right\}_{i=1}^{r}$, are Lipschitz continuous for a fixed $\bs{y} \in \mbbR^{M}$. In order to prove that  $\frac{\partial \Lagr_{c}(\bs{\theta}, \bs{y})}{\partial \theta_{i}}$ is $L_{c, \bs{y}}^{i}$-Lipschitz continuous, we need to bound the Lipschitz constant $L_{c, \bs{y}}^{i}$ from above. According to Lemma~\ref{lemma:lipmul}, we state the following:
\begin{align}
   L_{c, \bs{y}}^{i} & = \sqrt{r} \max_{j} \left \{\sup_{\bst} \left| \frac{\partial^{2} \Lagr_{c}(\bs{\theta}, \bs{y})}{\partial \theta_{i} \partial \theta_{j}} \right| \right \}_j\\ \nonumber
   & \leq \sqrt{r} \max_{j} \Bigg \{ \sup_{\bst} \Bigg \{ \lambda \left| \frac{\partial^{2} \langle I\otimes C \rangle_{\bst}}{\partial \theta_{i} \partial \theta_{j}} \right| + \lambda  \sum_{m=1}^{M} \left|y_{m}\right| \left|\frac{\partial^{2} \langle I\otimes A_{m} \rangle_{\bst}}{\partial \theta_{i} \partial \theta_{j}} \right|  \\ \nonumber
    & \qquad \qquad\qquad  + \lambda c  \sum_{m=1}^{M} \Bigg (  \left|b_{m}\right| \left|\frac{\partial^{2} \langle I\otimes A_{m} \rangle_{\bst}}{\partial \theta_{i} \partial \theta_{j}}\right| + \lambda \left| \langle I\otimes A_{m}\rangle_{\bst} \right| \left| \frac{\partial^{2} \langle I\otimes A_{m} \rangle_{\bst}}{\partial \theta_{i} \partial \theta_{j}} \right|\\ 
    & \qquad \qquad \qquad\qquad\qquad + \lambda \left| \frac{\partial \langle I\otimes A_{m} \rangle_{\bst}}{\partial \theta_{i}} \right| \left| \frac{\partial \langle I\otimes A_{m} \rangle_{\bst}}{\partial \theta_{j}} \right| \Bigg ) \Bigg \} \Bigg \}_j,
\end{align}
where $j \in \{1,\ldots, r\}$, and the aforementioned inequality follows from the triangle inequality. Let us first bound $\sup_{\bst} \left|\frac{\partial^{2} \langle I\otimes O \rangle_{\bst}}{\partial \theta_{i} \partial \theta_{j}}\right|$ from above as follows, where $O \in \{C, A_{1}, \dots, A_{M}\}$:
\begin{align}
& \sup_{\bst}\left |\frac{\partial^{2} \langle I\otimes O \rangle_{\bst}}{\partial \theta_{i} \partial \theta_{j}}\right| \\ \nonumber
&  \overset{\mathrm{(a)}}{=}\sup_{\bst} \left \{ \left| 2 \operatorname{Re} \Bigg[ \langle \bs{0} | U^{\dagger}(\bst) (I\otimes O) \frac{\partial^{2} U(\bst)}{\partial \theta_{i} \partial \theta_{j}} | \bs{0} \rangle \Bigg] + 2 \operatorname{Re} \Bigg[ \langle \bs{0} | \frac{\partial U^{\dagger}(\bst)}{\partial \theta_{i}} (I\otimes O) \frac{\partial U(\bst)}{\partial \theta_{j}} | \bs{0} \rangle \Bigg] \right | \right\} \\ \nonumber
&  \overset{\mathrm{(b)}}{\leq}2 \sup_{\bst} \left \{ \left|  \operatorname{Re} \Bigg[ \langle \bs{0} | U^{\dagger}(\bst) (I\otimes O) \frac{\partial^{2} U(\bst)}{\partial \theta_{i} \partial \theta_{j}} | \bs{0} \rangle \Bigg] \right|+ \left|  \operatorname{Re} \Bigg[ \langle \bs{0} | \frac{\partial U^{\dagger}(\bst)}{\partial \theta_{i}} (I\otimes O) \frac{\partial U(\bst)}{\partial \theta_{j}} | \bs{0} \rangle \Bigg] \right| \right\} \\ \nonumber
& \overset{\mathrm{(c)}}{\leq} 2 \sup_{\bst} \Bigg \{  \left|\langle \bs{0} | U^{\dagger}(\bst) (I\otimes O) U_{ff}(\bst)e^{-i\theta_{j} H_{j}}(-iH_{j})U_{fb}(\bst)e^{-i\theta_{i} H_{i}}(-iH_{i})U_{b}(\bst) | \bs{0} \rangle \right | \\ \nonumber
& \qquad\qquad\qquad +  \left|\langle \bs{0} |U_{b}^{\dagger}(\bst) e^{i\theta_{i} H_{i}}(iH_{i}) U_{f}^{\dagger}(\bst) (I\otimes O) U_{ff}(\bst)e^{-i\theta_{j} H_{j}}(-iH_{j})U_{b'}(\bst) | \bs{0} \rangle \right |\Bigg\} \\ 
&  \overset{\mathrm{(d)}}{\leq}  4  \lV O\rV  \lV H_{j}\rV \lV H_{i}\rV, \label{eq:sdexpbound}
\end{align}
where equality (a) follows from the chain rule applied twice (see~\eqref{eq:chainruleexp}). Inequality (b) follows from the triangle inequality.
Inequality (c) follows from the fact that $\left|\operatorname{Re}[z]\right| \leq \left| z \right|$ for all  $z \in \mathbb{C}$, as well as from the assumption that $U(\bst)$ has the following decomposition:
\begin{equation}
   U(\bst) = U_{ff}(\bst)e^{-i\theta_{j}H_{j}}U_{fb}(\bst)e^{-i\theta_{i}H_{i}}U_{b}(\bst).
\end{equation}
Additionally, let  $U_{ff}(\bst)e^{-i\theta_{j}H_{j}}U_{fb}(\bst) =  U_{f}(\bst)$ and $U_{fb}(\bst)e^{-i\theta_{i}H_{i}}U_{b}(\bst) = U_{b'}(\bst)$.  Inequality (d) follows from a set of arguments similar to those in the proof of Lemma~\ref{lemma:lipconexp}.

Second, we bound $\sup_{\bst} \left| \langle I\otimes O \rangle_{\bst}\right|$ from above as follows, where $O \in \{A_{1}, \dots, A_{M}\}$:
\begin{align}
    \sup_{\bst} \left| \langle I\otimes O \rangle_{\bst}\right| & = \sup_{\bst} \left| \langle \bs{0} | U^{\dagger}(\bst) (I \otimes O ) U(\bst) | \bs{0} \rangle \right| \\ 
    & \leq \lV O \rV.\label{eq:expbound}
\end{align}

Finally, we bound $\sup_{\bst} \left| \frac{\partial \langle I\otimes O \rangle_{\bst}}{\partial \theta_{i}} \right|$ from above as follows, where $O \in \{A_{1}, \dots, A_{M}\}$:
\begin{align}
    \sup_{\bst} \left| \frac{\partial \langle I\otimes O \rangle_{\bst}}{\partial \theta_{i}} \right| & = \sup_{\bst} \left| 2 \operatorname{Re} \Bigg[ \langle \bs{0} | U^{\dagger}(\bst) (I \otimes O) \frac{\partial U(\bst)}{\partial \theta_{i}} | \bs{0} \rangle \Bigg] \right | \\ 
    & \leq 2 \lV O \rV \lV H_{i}\rV.\label{eq:fdexpbound}
\end{align}
The above inequality follows from the similar set of arguments made in the proof of Lemma~\ref{lemma:lipconexp}. 

Note that the upper bounds in~\eqref{eq:sdexpbound},~\eqref{eq:expbound}, and~\eqref{eq:fdexpbound} are independent of $\bst$. Hence, using these upper bounds, we bound $L_{c, \bs{y}}^{i}$ from above as follows:
\begin{align}
   L_{c, \bs{y}}^{i}  & \leq \sqrt{r} \max_{j} \Bigg\{4\lambda  \lV C\rV \lV H_{j}\rV\lV H_{i}\rV+  \lambda \sum_{m=1}^{M} 4 \left|y_{m}\right|  \lV A_{m}\rV \lV H_{j}\rV\lV H_{i}\rV \\ \nonumber
    & \quad  + \lambda c  \sum_{m=1}^{M}\left( 4 \left|b_{m}\right| \lV A_{m}\rV \lV H_{j}\rV\lV H_{i}\rV + 4\lambda  \lV A_{m}\rV^{2}  \lV H_{j}\rV\lV H_{i}\rV+ 4\lambda  \lV A_{m}\rV^{2}  \lV H_{j}\rV\lV H_{i}\rV \right)\Bigg\}_j\\
    & = 4 \lambda \sqrt{r} \lV H_{i}\rV   \left(\lV C\rV + \sum_{m=1}^{M}  \left(\left(\left|y_{m}\right| + c \left|b_{m}\right|\right) \lV A_{m}\rV +  2 c \lambda \lV A_{m}\rV^{2} \right) \right) \max_{j} \left\{\lV H_{j}\rV\right\}_j \label{eq:Libound}
\end{align}

From~\eqref{eq:Libound} we conclude that the Lipschitz constant $L_{c, \bs{y}}^{i}$ is bounded from above by a positive number. Hence, the Lipschitz constant $L_{c, \bs{y}}$ of $\nabla_{\bs{\theta}}\Lagr_{c}(\bs{\theta}, \bs{y})$ is also bounded from above by a positive number because $L_{c, \bs{y}} = \sqrt{\sum_{i=1}^{r} L_{c, \bs{y}}^{i}}$ according to Lemma~\ref{lemma:lipvec}. Thus, the function $\Lagr_{c}(\bs{\theta}, \bs{y})$ is $L_{c, \bs{y}}$-smooth for a fixed $\bsy \in \mathbb{R}^{M}$.

\subsection{Proof of Lemma~\ref{lemma:smoothnessF}}

\label{app:smoothnessF}

The multivariate vector-valued function $\nabla_{\bs{\theta}}\mathcal{F}_{\gamma}(\bs{\theta}, \bar{\bs{y}})$ for a fixed $\bar{\bs{y}} \geq 0$, is $L_{\gamma, \bar{\bs{y}}}$-Lipschitz continuous if all its components, i.e.,  $\left\{\frac{\partial \mathcal{F}_{\gamma}(\bs{\theta}, \bar{\bs{y}})}{\partial \theta_{i}} \right\}_{i=1}^{r}$, are Lipschitz continuous for a fixed $\bar{\bs{y}} \geq 0$. In order to prove that  $\frac{\partial \mathcal{F}_{\gamma}(\bs{\theta}, \bar{\bs{y}})}{\partial \theta_{i}}$ is $L_{\gamma, \bar{\bs{y}}}^{i}$-Lipschitz continuous, we first bound the Lipschitz constant $L_{\gamma, \bar{\bs{y}}}^{i}$ from above as follows:
\begin{align}
   L_{\gamma, \bar{\bs{y}}}^{i} & = \sqrt{r} \max_{j} \left \{\sup_{\bst} \left| \frac{\partial^{2} \mathcal{F}_{\gamma}(\bs{\theta}, \bar{\bs{y}})}{\partial \theta_{i} \partial \theta_{j}} \right| \right \}_j\\ 
   \begin{split}
       & = b_{M} \sqrt{r} \max_{j} \Bigg \{\sup_{\bst} \Bigg| \frac{e^{\gamma \langle I \otimes H(\bar{\bs{y}}) \rangle_{\bst}}}{\left(e^{\gamma \langle I \otimes H(\bar{\bs{y}}) \rangle_{\bst}} + 1 \right)} \frac{\partial^{2} \langle I \otimes H(\bar{\bs{y}}) \rangle_{\bst}}{\partial \theta_{i} \partial \theta_{j}} \\
       &\qquad \qquad\qquad\qquad + \frac{e^{\gamma \langle I \otimes H(\bar{\bs{y}}) \rangle_{\bst}}}{\left(e^{\gamma \langle I \otimes H(\bar{\bs{y}}) \rangle_{\bst}} + 1 \right)^2} \gamma \frac{\partial \langle I \otimes H(\bar{\bs{y}}) \rangle_{\bst}}{\partial \theta_{i}} \frac{\partial \langle I \otimes H(\bar{\bs{y}}) \rangle_{\bst}}{\partial \theta_{j}} \Bigg| \Bigg \}_j
   \end{split}\\ 
   & \leq b_{M} \sqrt{r} \max_{j} \left \{\sup_{\bst} \left| \frac{\partial^{2} \langle I \otimes H(\bar{\bs{y}}) \rangle_{\bst}}{\partial \theta_{i} \partial \theta_{j}} + \gamma \frac{\partial \langle I \otimes H(\bar{\bs{y}}) \rangle_{\bst}}{\partial \theta_{i}} \frac{\partial \langle I \otimes H(\bar{\bs{y}}) \rangle_{\bst}}{\partial \theta_{j}} \right| \right \}_j\\ 
   & \leq b_{M} \sqrt{r} \max_{j} \left \{\sup_{\bst} \left| \frac{\partial^{2} \langle I \otimes H(\bar{\bs{y}}) \rangle_{\bst}}{\partial \theta_{i} \partial \theta_{j}} \right| +  \gamma \left| \frac{\partial \langle I \otimes H(\bar{\bs{y}}) \rangle_{\bst}}{\partial \theta_{i}} \right| \left|\frac{\partial \langle I \otimes H(\bar{\bs{y}}) \rangle_{\bst}}{\partial \theta_{j}} \right| \right \}_j\\ 
   & \leq 4 b_{M} \sqrt{r} \lV H_{i}\rV  \left(\lV H(\bar{\bs{y}})\rV + 4 \gamma \lV H(\bar{\bs{y}})\rV^{2}\right)\max_{j}\left\{\lV H_{j}\rV\right\}_j.
\end{align}
The last inequality follows from~\eqref{eq:fdexpbound} and~\eqref{eq:sdexpbound}.
We see that $L_{\gamma, \bar{\bs{y}}}^{i}$ is bounded from above by a positive number for a fixed $\bar{\bs{y}}$. Therefore, the Lipschitz constant $L_{\gamma, \bar{\bs{y}}}$ is also bounded from above by a positve number because according to Lemma~\ref{lemma:lipvec}, we have $L_{\gamma, \bar{\bs{y}}} = \sqrt{\sum_{i=1}^{r} L_{\gamma, \bar{\bs{y}}}^{i}}$. Thus, function $\mathcal{F}_{\gamma}(\bs{\theta}, \bar{\bs{y}})$ is $L_{\gamma, \bar{\bs{y}}}$-smooth for a fixed $\bar{\bs{y}} \geq 0$.
\end{document}